\documentclass[letterpaper,11pt]{article}

\usepackage{amsthm,amsmath,amssymb,amsfonts}
\usepackage[american]{babel}
\usepackage{graphicx}
\usepackage[giveninits=true,maxbibnames=99,style=alphabetic,maxalphanames=4,minalphanames=3,backend=biber]{biblatex}
\usepackage{tailwindcolors}
\usepackage{hyperref}
\usepackage{xurl}
\usepackage[capitalise]{cleveref}
\hypersetup{
  breaklinks=true,
  pdfusetitle=true,
  hypertexnames=false,
  colorlinks   = true, %Colours links instead of ugly boxes
  urlcolor     = blue800, %Colour for external hyperlinks
  linkcolor    = green700, %Colour of internal links
  citecolor    = blue900 %Colour of citations
}
\usepackage{mathtools}
\usepackage{enumitem}
\usepackage{csquotes}
\usepackage[noend]{algpseudocode}
\usepackage{algorithm}
\usepackage{xparse}
\usepackage{xspace}
\usepackage{color}
\usepackage{caption}
\usepackage{subcaption}
\usepackage{fullpage}
\usepackage{placeins}
\usepackage{xcolor}
\usepackage{makecell}
\usepackage{thmtools}
\usepackage{thm-restate}
\usepackage{tikz}
\usepackage{nicematrix}
\usepackage{mathrsfs}
\usepackage[equations=false]{resizegather}

\usepackage[paper=letterpaper,margin=1in]{geometry}

\mathchardef\mhyphen="2D % Define a "math hyphen"

\newcommand\newmathabbrev[2]{\newcommand{#1}{\ensuremath{#2}\xspace}}

\renewcommand{\exp}{\mathrm{exp}}

\newcommand\cfont\mathsf
\newmathabbrev\p{\cfont{P}}
\newmathabbrev{\NN}{\mathbb N}
\newmathabbrev{\DD}{\mathbb D}
\newmathabbrev{\RR}{\mathbb R}
\newmathabbrev{\FF}{\mathbb F}
\newmathabbrev{\KK}{\mathbb K}
\newmathabbrev\NP{\cfont{NP}}
\newmathabbrev\DTIME{\cfont{DTIME}}
\newmathabbrev\tSAT{3\cfont{\mhyphen{}SAT}}
\newcommand\MAXkSAT[1]{\cfont{MAX}\mhyphen{}#1\mhyphen{}\cfont{SAT}}
\newmathabbrev\MA{\cfont{MA}}
\newmathabbrev\AM{\cfont{AM}}
\newmathabbrev\NPDAG{\cfont{NP\mhyphen{}DAG}}
\newmathabbrev\QMADAG{\cfont{QMA\mhyphen{}DAG}}
\newmathabbrev\yes{\mathrm{yes}}
\newmathabbrev\no{\mathrm{no}}
\newmathabbrev\US{\cfont{US}}
\newmathabbrev\NC{\cfont{NC}}
\newmathabbrev\FP{\cfont{FP}}
\newmathabbrev\PP{\cfont{PP}}
\newmathabbrev\CeP{\cfont{C_=P}}
\newmathabbrev\coCeP{\cfont{coC_=P}}
\newmathabbrev\PH{\cfont{PH}}
\newmathabbrev\SAT{\cfont{SAT}}
\newmathabbrev\kSAT{k\mhyphen\cfont{SAT}}
\newmathabbrev\QSAT{\cfont{QSAT}}
\newmathabbrev\hQSAT{\mhyphen\QSAT}
\newmathabbrev\SPP{\cfont{SPP}}
\newmathabbrev\GapP{\cfont{GapP}}
\newmathabbrev\BQP{\cfont{BQP}}
\newmathabbrev\BQPo{\BQP_1}
\newmathabbrev\QP{\cfont{QP}}
\newmathabbrev\StoqMA{\cfont{StoqMA}}
\newmathabbrev\coNP{\cfont{coNP}}
\newmathabbrev\AzPP{\cfont{A_0PP}}
\newmathabbrev\QMA{\cfont{QMA}}
\newmathabbrev\QMAo{\QMA_1}
\newmathabbrev\coQMA{\cfont{coQMA}}
\newmathabbrev\BPP{\cfont{BPP}}
\newmathabbrev\QCMA{\cfont{QCMA}}
\newmathabbrev\pNPlog{\p^{\NP[\log]}}
\newmathabbrev\pNP{\p^{\NP}}
\newmathabbrev\pNPtwo{\p^{\NP[2]}}
\newmathabbrev\pNPone{\p^{\NP[1]}}
\newmathabbrev\pParSAT{\p^{||\SAT}}
\newmathabbrev\pQMApar{\p^{||\QMA}}
\newmathabbrev\pCpar{\p^{||\C}}
\newmathabbrev\pStoqMApar{\p^{||\StoqMA}}
\newmathabbrev\pQMAlog{\p^{\QMA[\log]}}
\newmathabbrev\pClog{\p^{\textup{C}[\log]}}
\newmathabbrev\pC{\p^{\textup{C}}}
\newmathabbrev\QMASPACE{\cfont{QMASPACE}}
\newmathabbrev\pQMAtlog{\p^{\QMA(2)[\log]}}
\newmathabbrev\pStoqMAlog{\p^{\StoqMA[\log]}}
\newmathabbrev\pQMApt{\p^{\Vert\QMA(2)}}
\newmathabbrev\pQMA{\p^{\QMA}}
\newmathabbrev\SharpP{\cfont{\#P}}
\newmathabbrev\pSharP{\p^{\SharpP[1]}}
\newmathabbrev\PromisePP{\cfont{PromisePP}}
\newmathabbrev\lett{\le_\mathrm{tt}}
\newmathabbrev\YES{\mathsf{YES}}
\newmathabbrev\NO{\mathsf{NO}}
\newmathabbrev\PSPACE{\cfont{PSPACE}}
\newmathabbrev\IP{\cfont{IP}}
\newmathabbrev\POLY{\cfont{POLY}}
\newmathabbrev\DAG{\cfont{DAG}}
\newmathabbrev\StoqMADAG{\StoqMA\mhyphen\cfont{DAG}}
\newmathabbrev\CDAG{C\mhyphen\cfont{DAG}}
\newmathabbrev\CDAGf{C\mhyphen\cfont{DAG}_f}
\newmathabbrev\CDAGs{C\mhyphen\cfont{DAG}_s}
\newmathabbrev\CDAGd{C\mhyphen\cfont{DAG}_{d}}
\newmathabbrev\CDAGo{C\mhyphen\cfont{DAG}_1}
\newmathabbrev\LOGS{\cfont{LOGS}}
\newmathabbrev\TAUT{\cfont{TAUTOLOGY}}
\newmathabbrev\SBQP{\cfont{SBQP}}
\newmathabbrev\Fc{F_\coNP}
\newmathabbrev\Fa{F_\AzPP}
\newmathabbrev\GSCON{\cfont{GSCON}}
\newmathabbrev\GSCONexp{\GSCON_\cfont{exp}}
\newmathabbrev\QMAexp{\QMA_\cfont{exp}}
\newmathabbrev\UQMA{\cfont{UQMA}}
\newmathabbrev\R{\mathbb R}
\newmathabbrev\Trees{\cfont{TREES}}
\newmathabbrev\apxsim{\cfont{APX\mhyphen{}SIM}}
\newmathabbrev\AWPP{\cfont{AWPP}}
\newmathabbrev\X{\mathcal{X}}
\newmathabbrev\calG{\mathcal{G}}
\newmathabbrev\calK{\mathcal{K}}
\newmathabbrev\calJ{\mathcal{J}}
\newmathabbrev\calC{\mathcal{C}}
\newcommand{\CK}{\calC^k(\calK)}
\newcommand{\CKk}[1]{\calC^{#1}(\calK)}
\newcommand{\HP}{\cfont{HP}}

\newmathabbrev\calA{\mathcal{A}}
\newmathabbrev\calB{\mathcal{B}}
\newmathabbrev\Y{\mathcal{Y}}
\renewcommand\H{\ensuremath{\mathcal{H}}}
\newmathabbrev\Z{\mathcal{Z}}
\newmathabbrev\ZZ{\mathbb{Z}}
\newcommand\Hprop{H_\mathrm{prop}}

\newcommand\Hclock{H_{\mathrm{clock}}}

\newcommand\Jclock{J_{\mathrm{clock}}}
\newcommand\Sclock{\mathscr{C}}

\newcommand\Hin{H_\mathrm{in}}
\newcommand\Hsplit{H_\mathrm{split}}

\newcommand\Hout{H_\mathrm{out}}

\newmathabbrev\A{\mathcal{A}}
\newmathabbrev\rmU{\mathrm{U}}
\newmathabbrev\rmO{\mathrm{O}}

\newcommand\kLH{k\cfont{\mhyphen{}LH}}
\newcommand\kQSAT{k\cfont{\mhyphen{}QSAT}}
\newcommand\kELH{k\cfont{\mhyphen{}ELH}}

\newcommand\lELH{l\cfont{\mhyphen{}ELH}}

\newcommand\hLH{\cfont{\mhyphen{}LH}}
\newcommand\hELH{\cfont{\mhyphen{}ELH}}

\newcommand\AESSH{\cfont{AESSH}}
\newcommand\SH{\cfont{SH}}

\newcommand\SSH{\cfont{SSH}}
\newcommand\ESH{\cfont{ESH}}
\newcommand\ESSH{\cfont{ESSH}}
\newcommand\CH{\cfont{CH}}
\newcommand\GCH{\cfont{GCH}}

\newcommand{\wtU}{\widetilde{U}}

\newcommand{\whU}{\widehat{U}}
\newcommand{\wht}{\widehat{t}}
\newcommand{\wh}[1]{\widehat{#1}}

\newcommand{\whV}{\widehat{V}}
\newcommand{\whpsi}{\widehat{\psi}}
\newcommand{\whphi}{\widehat{\phi}}

\newcommand{\psiinit}{\psi_{\mathrm{init}}}

\newcommand{\whpsiinitp}[1]{\whpsi_{#1,\mathrm{init}}}
\newcommand{\whpsiinit}{\whpsi_{\mathrm{init}}}
\newcommand\B{\mathcal B}

\newcommand\Pirej{\Pi_\mathrm{rej}}

\newmathabbrev\DAGSSAT{\DAGS(\SAT)}
\newmathabbrev\DAGS{\mathrm{DAGS}}
\newmathabbrev\DAGSNP{\DAGS(\NP)}

\newmathabbrev\AND{\cfont{AND}}

\newmathabbrev\STCONN{{S,T}\cfont{\mhyphen{}CONN}}
\newmathabbrev\CNF{\cfont{CNF}}
\newmathabbrev\NEXP{\cfont{NEXP}}
\newmathabbrev\NPSPACE{\cfont{NPSPACE}}
\newmathabbrev\QCMASPACE{\cfont{QCMASPACE}}
\newmathabbrev\BQPSPACE{\cfont{BQPSPACE}}
\newmathabbrev{\PCP}{\cfont{PCP}}
\newmathabbrev\BQUPSPACE{\cfont{BQ_UPSPACE}}
\newmathabbrev\QMAt{\QMA(2)}
\newmathabbrev\QMAtexp{\QMAt_{\exp}}
\newmathabbrev\MIP{\cfont{MIP}}
\newmathabbrev\QMIP{\cfont{MIP}}
\newmathabbrev\QSZK{\cfont{QSZK}}
\newmathabbrev\QIP{\cfont{QIP}}
\newmathabbrev\MIPt{\MIP(2)}
\newmathabbrev\BellQMA{\cfont{BellQMA}}
\newmathabbrev\BellQMAt{\BellQMA(2)}
\newmathabbrev\BellQMAtexp{\BellQMAt_{\exp}}
\newmathabbrev\Upyth{U_{\mathrm{Pyth.}}}

\newmathabbrev\CNOT{\mathsf{CNOT}}
\newmathabbrev\Hg{\mathsf{H}}
\newmathabbrev\Rg{\mathsf{R}}
\newmathabbrev\Htg{\mathsf{\widetilde{H}}}
\newmathabbrev\CSg{\mathsf{CS}}
\newmathabbrev\CXg{\mathsf{CX}}
\newmathabbrev\CZg{\mathsf{CZ}}
\newmathabbrev\CCXg{\mathsf{CCX}}
\newmathabbrev\CCZg{\mathsf{CCZ}}
\newmathabbrev\Sg{\mathsf{S}}
\newmathabbrev\Xg{\mathsf{X}}
\newmathabbrev\Yg{\mathsf{Y}}
\newmathabbrev\Zg{\mathsf{Z}}
\newmathabbrev\Tg{\mathsf{T}}
\newmathabbrev\Bg{\mathsf{B}}
\newmathabbrev\Ig{\mathsf{I}}

\newcommand{\calH}{\mathcal{H}}
\newcommand{\calS}{\mathcal{S}}

\newcommand{\iu}{\mathrm{i}}

\protected\def\verythinspace{%
  \ifmmode
    \mskip0.5\thinmuskip
  \else
    \ifhmode
      \kern0.08334em
    \fi
  \fi
}

\newcommand{\bin}{\{0,1\}}

\newcommand{\CC}{\mathbb C}
\newcommand{\QQ}{\mathbb Q}

\newcommand{\be}{\begin{equation}}
\newcommand{\ee}{\end{equation}}

\newcommand{\up}{\mathrm{up}}
\newcommand{\unitary}{\textup{U}}

\renewcommand{\epsilon}{\varepsilon}

\DeclareMathOperator{\Tr}{Tr}

\DeclareMathOperator{\sgn}{sgn}

\newcommand{\Null}{\mathcal{N}}
\DeclareMathOperator{\Image}{Im}
\DeclareMathOperator{\Kernel}{Ker}

\newcommand\lmin{\lambda_{\mathrm{min}}}

\newcommand{\Cl}{\mathrm{Cl}}

\newcommand{\poly}{\mathrm{poly}}

\DeclareMathOperator{\Span}{Span}

\DeclareMathOperator{\corank}{corank}

\DeclarePairedDelimiter\bra{\langle}{\rvert}
\DeclarePairedDelimiter\ket{\lvert}{\rangle}

\DeclarePairedDelimiter\abs{\lvert}{\rvert}
\DeclarePairedDelimiter\norm{\lVert}{\rVert}

\DeclarePairedDelimiter\fnorm{\lVert}{\rVert_{\mathrm F}}

\DeclarePairedDelimiterX\braket[2]{\langle}{\rangle}{#1 \delimsize\vert #2}
\DeclarePairedDelimiterX\ketbra[2]{\lvert}{\rvert}{#1 \delimsize\rangle\delimsize\langle #2}

\newcommand{\braketb}[2]{\bra{#1}#2\ket{#1}}
\newcommand{\braketc}[1]{\braket{#1}{#1}}
\newcommand{\ketbraa}[1]{{#1 \renewcommand\ket\bra #1}}
\newcommand{\ketbrab}[1]{\ketbra{#1}{#1}}

\setlist[itemize]{noitemsep, topsep=0pt}
\setlist[enumerate]{noitemsep, topsep=0pt}

\crefname{claim}{Claim}{Claims}
\Crefname{claim}{Claim}{Claims}

\declaretheorem[numberwithin=section]{theorem}

\declaretheorem[sibling=theorem]{lemma}
\declaretheorem[sibling=theorem]{claim}
\declaretheorem[sibling=theorem]{proposition}

\declaretheorem[sibling=theorem,style=definition]{problem}
\declaretheorem[sibling=theorem,style=definition]{definition}
\declaretheorem[sibling=theorem,style=definition]{remark}

\newcommand{\pacc}{p_{\textup{acc}}}
\newcommand{\prej}{p_{\textup{rej}}}

\newcommand{\Ayes}{A_{\textup{yes}}} %CHECK
\newcommand{\Ano}{A_{\textup{no}}} %CHECK
\newcommand{\psihist}{\psi_{\textup{hist}}}

\makeatletter
\newcommand{\subalign}[1]{%
  \vcenter{%
    \Let@ \restore@math@cr \default@tag
    \baselineskip\fontdimen10 \scriptfont\tw@
    \advance\baselineskip\fontdimen12 \scriptfont\tw@
    \lineskip\thr@@\fontdimen8 \scriptfont\thr@@
    \lineskiplimit\lineskip
    \ialign{\hfil$\m@th\scriptstyle##$&$\m@th\scriptstyle{}##$\hfil\crcr
      #1\crcr
    }%
  }%
}
\NewDocumentCommand{\LeftComment}{s m}{%
  \Statex \IfBooleanF{#1}{\hspace*{\ALG@thistlm}}\(\triangleright\) #2}

\def\moverlay{\mathpalette\mov@rlay}
\def\mov@rlay#1#2{\leavevmode\vtop{%
   \baselineskip\z@skip \lineskiplimit-\maxdimen
   \ialign{\hfil$\m@th#1##$\hfil\cr#2\crcr}}}
\newcommand{\charfusion}[3][\mathord]{
    #1{\ifx#1\mathop\vphantom{#2}\fi
        \mathpalette\mov@rlay{#2\cr#3}
      }
    \ifx#1\mathop\expandafter\displaylimits\fi}

\makeatother

\algnewcommand{\LineComment}[1]{\State \(\triangleright\) #1}

\algblockdefx[ON]{Blk}{EndBlk}[1]
  {#1}
  {}

\makeatletter
\ifthenelse{\equal{\ALG@noend}{t}}%
  {\algtext*{EndBlk}}
  {}%
\makeatother

\newcolumntype{M}[1]{>{\centering\arraybackslash$}m{#1}<{$}}

\makeatletter
\newcommand\@rcolwidth{0.67em}

\def\@rarray[#1]{\arraycolsep=0pt\array{*\c@MaxMatrixCols {M{#1}}}}
\makeatother

\newcommand{\splitatcommas}[1]{%https://tex.stackexchange.com/a/309558
  \begingroup
  \begingroup\lccode`~=`, \lowercase{\endgroup
    \edef~{\mathchar\the\mathcode`, \penalty0 \noexpand\hspace{0pt plus 1em}}%
  }\mathcode`,="8000 #1%
  \endgroup
}

\AtEveryBibitem{%
  \clearlist{language}%
}

\usetikzlibrary{positioning,quotes,arrows.meta}

\title{Towards a universal gateset for \texorpdfstring{QMA\textsubscript{1}}{QMA\_1}}
\author{Dorian Rudolph\footnote{Department of Computer Science and Institute for Photonic Quantum Systems (PhoQS), Paderborn University, Germany. Email: dorian.rudolph@upb.de}}
\date{April 11, 2024}
\begin{document}

\maketitle

\begin{abstract}
$\mathsf{QMA}_1$ is $\mathsf{QMA}$ with perfect completeness, i.e., the prover must accept with a probability of \emph{exactly} $1$ in the YES-case.
Whether $\mathsf{QMA}_1$ and $\mathsf{QMA}$ are equal is still a major open problem (classically, we have $\mathsf{MA}_1=\mathsf{MA}$), and only a quantum oracle separation due to Aaronson (QIC 2009) is known.
Furthermore, $\mathsf{QMA}_1$ does not actually have a single agreed-upon definition, since it depends on the choice of gateset, and the Solovay-Kitaev theorem only approximately synthesizes arbitrary gates from a universal gateset, not necessarily preserving perfect completeness.
For a gateset $\mathcal{G}$, we define $\mathsf{QMA}_1^{\mathcal{G}}$ as $\mathsf{QMA}_1$ restricted to verifiers only using gates of $\mathcal{G}$.
Generally, it is not clear at all what the relationship of $\mathsf{QMA}_1^{\mathcal{G}}$ and $\mathsf{QMA}_1^{\mathcal{G}'}$, even for two universal gatesets $\mathcal{G}$ and $\mathcal{G}'$.
Therefore, the $\mathsf{QMA}_1^{\mathcal{G}}$ classes form a potentially infinite hierarchy!

In this paper, we bring some structure to this chaos by proving that for each $k\in\mathbb{N}$, there exists a universal gateset $\mathcal{G}_{2^k}$ (defined by Amy, Glaudell, Kelso, Maxwell, Mendelson, Ross, Reversible Computation, 2024), so that $\mathsf{QMA}_1^{\mathcal{G}} \subseteq \mathsf{QMA}_1^{\mathcal{G}_{2^k}}$ for all gatesets $\mathcal{G}$ consisting of unitaries in the cyclotomic field $\mathbb{Q}(\zeta_{2^k})$, where $\zeta_{2^k} = e^{2\pi\mathrm{i}/2^k}$ is a primitive $2^k$-th root of unity.
For $\mathsf{BQP}_1$, we can even show that $\mathcal{G}_2$ suffices for all $2^k$-th cyclotomic fields, i.e., $\mathsf{BQP}_1^{\mathcal{G}}\subseteq\mathsf{BQP}_1^{\mathcal{G}_2}$ for all $\mathcal{G}$ in $\mathbb{Q}(\zeta_{2^k})$.
We exhibit complete problems for all $\mathsf{QMA}_1^{\mathcal{G}_{2^k}}$:
Quantum $l$-SAT in $\mathbb{Q}(\zeta_{2^k})$ is complete for $\mathsf{QMA}_1^{\mathcal{G}_{2^k}}$ for all $l\ge4$, and $l=3$ if $k\ge3$, where quantum $l$-SAT is the problem of deciding whether a set of $l$-local Hamiltonians has a common ground state.

Our techniques rely on representing operators as linear combinations of unitaries, which was pioneered by Childs and Wiebe (QIC 2012), and allows us to exactly apply even non-unitary operators to a quantum state, using postselection.
We show the first $\mathsf{QMA}_1$-complete $2$-local Hamiltonian problem: It is $\mathsf{QMA}_1^{\mathcal{G}_{2^k}}$-complete (for $k\ge3$) to decide whether a given $2$-local Hamiltonian $H$ in $\mathbb{Q}(\zeta_{2^k})$ has a nonempty nullspace (i.e. $\sigma_1(H)=0$) or $\sigma_1(H)\ge1/\mathrm{poly}$, where $\sigma_1$ denotes the smallest singular value.
Our techniques also extend to sparse Hamiltonians, and so we can prove the first $\mathsf{QMA}_1(2)$-complete (i.e. $\mathsf{QMA}_1$ with two unentangled Merlins) Hamiltonian problem, which is a variant of the separable sparse Hamiltonian problem (Chailloux and Sattath, CCC 2012).
Finally, we prove that the Gapped Clique Homology problem on weighted graphs defined by King and Kohler (FOCS 2024) is $\mathsf{QMA}_1^{\mathcal{G}_2}$-complete, and the Clique Homology problem (Kaibel and Pfetsch, 2002) without promise gap is $\mathsf{PSPACE}$-complete, resolving a conjecture of Crichigno and Kohler (Nat. Commun. 2024).
\end{abstract}

\newpage
\section{Introduction}

The complexity class $\QMAo$ was introduced by Bravyi in 2006 \cite{Bra06} as $\QMA$ with one-sided error (or \emph{perfect completeness}), i.e., in the YES-case there exists a proof that the verifier accepts with a probability of \emph{exactly} $1$.
Bravyi shows that the Quantum $k$-SAT problem ($\kQSAT$) is $\QMAo$-complete for $k\ge4$ and in $\p$ for $k=2$, where $\kQSAT$ is the problem of deciding whether a given collection of $k$-local Hamiltonians has a common ground state.
$\kQSAT$ can be seen as a quantum analogue of the classical $\kSAT$ problem, which is $\NP$-complete for $k\ge3$ \cite{Coo71,Lev73,Karp72}, and in $\p$ for $k=2$ \cite{Qui59,DP60,Kro67,EIS76,APT79,Pap91}.
Later works even show that $2\hQSAT$ is solvable in linear time~\cite{ASSZ16,BG16}, and $3\hQSAT$ is $\QMAo$-complete \cite{GN13}.
Quantum $2$-SAT on qu\emph{d}its is $\QMAo$-complete \cite{ER08,Nag08,RGN24}, with the most recent result being the $\QMAo$-completeness of $(3,4)\hQSAT$ and $(2,5)\hQSAT$ \cite{RGN24}, where in $(k,l)\hQSAT$ each local term acts on one qu-$k$-it and one qu-$l$-it.
$\kQSAT$ is a special case of the more well-known $k$-local Hamiltonian problem ($\kLH$), where the goal is to approximate $\lmin(H)$ for a $k$-local Hamiltonian $H$.
$\kLH$ can be seen as a quantum analogue of the classical $\MAXkSAT{k}$ problem (i.e. what is the maximum number of simultaneously satisfiable constraints in a $2$-CNF formula?).
The $2$-local Hamiltonian problem is $\QMA$-complete \cite{KSV02,KKR05,CM16}, just as $\MAXkSAT{2}$ is $\NP$-complete \cite{Gar76}.

An interesting question to ask is now whether quantum $k$-SAT is also $\QMA$-complete, or in other words, is $\QMA=\QMAo$?
Classically, we have $\MA=\MA_1$ \cite{ZF87,GZ11}.
Quantum interactive proof systems ($\QIP$) also have perfect completeness \cite{KW00,KLN15}.
$\QMA$ even has a two-message quantum interactive proof system with perfect completeness, i.e., $\QMA\in\QIP_1(2)$ \cite{KLN15}.
\cite{JKNN12} have shown that \QCMA (i.e., $\QMA$ with classical proofs) has perfect completeness, i.e., $\QCMA\subseteq\QMAo$.
Despite these positive results, the question whether $\QMA=\QMAo$ remains open to this day.
Aaronson~\cite{Aar09} gives some negative evidence in the form of a \emph{quantum} oracle separation between $\QMA$ and $\QMAo$.
The idea is to give the verifier access to an oracle performing a rotation by some angle $\theta$, and the problem is to distinguish the cases $\theta\in[1,2]$ (YES-case) or $\theta=0$ (NO-case).
Aaronson uses techniques from real analysis to prove that if the verifier accepts with probability $1$ for all $\theta$ in an open set, then it must accept for all $\theta$.

So far we have not touched upon the ``gateset issue''.
That is, the definition of $\QMAo$ depends on which gateset the verifier uses.
We write $\QMAo^\calG$ to denote $\QMAo$ with gateset $\calG$, where $\calG$ is a set of unitaries.
Notably, we cannot use the Solovay-Kitaev algorithm \cite{Kit97,DN05,BG21} to synthesize gates because that only approximates gates, not necessarily preserving perfect completeness.
Therefore, we do not in general know whether $\QMAo^{\calG} = \QMAo^{\calG'}$ for two different universal gatesets $\calG$ and $\calG'$.
Thus, the $\QMAo^\calG$ complexity classes form a potentially infinite hierarchy!

Bravyi~\cite{Bra06} defines $\QMAo$ as $\QMAo^{\calG}$, where $\calG$ is the set of all $3$-qubit unitaries in some subfield $\FF\subseteq\CC$ with exact representation.
The corresponding complete $4\hQSAT$ problem also allows all projectors with elements in $\FF$.
Thus, $\kQSAT\in\QMAo$ requires the verifier to be constructed specifically for the instance.\footnote{We remark that there is a slight issue with this procedure (see \cite[Lemma 5]{Bra06}), also mentioned by Gosset and Nagaj \cite{GN13}. Decomposing general unitaries into $2$-level unitaries (see \cite{NC10}) does not necessary produce unitaries inside $\FF$ (e.g. for $\FF=\QQ$).}
We cannot just give a classical description of the Hamiltonian to the quantum verifier.

Gosset and Nagaj~\cite{GN13} instead only consider the gateset $\{\Hg,\CXg,\Tg\}$ (Hadamard, CNOT, T-gate), and show completeness for $\kQSAT$, where each local projector $\Pi$ is in $\ZZ[\frac1{\sqrt2},\iu]$,\footnote{We use `$\iu$' to denote the imaginary unit, disambiguating it from the index `$i$'.} or there exists a unitary with only these elements, such that 
\begin{equation*}
  U\Pi U^\dagger = \ketbraa{(\sqrt{1/3}\ket{000}-\sqrt{2/3}\ket{001})}.
\end{equation*}
That definition of $\kQSAT$ unfortunately appears somewhat ``unnatural''.
% In this work we will ``bridge the gap'' between the definitions of \cite{Bra06} and \cite{GN13} in that we show completeness for $\kQSAT$ in a \emph{field} for $\QMAo^\calG$ with a \emph{finite} gateset $\calG$.
In this work, we will ``bridge the gap'' between the definitions of \cite{Bra06} and \cite{GN13} by proving completeness for $\kQSAT$ in \emph{fields} $\FF\subseteq\CC$ for $\QMAo^\calG$ with \emph{finite} gatesets $\calG$.

The containment result for $\kQSAT\in\QMAo$ in \cite{GN13} relies on Giles and Selinger's \cite{GS13} algorithm for \emph{exactly} synthesizing any $n$-qubit unitary with entries in the ring $\ZZ[1/\sqrt{2},\iu]$ with the ``Clifford + T'' gateset.
Their work has been extended to further gatesets by \cite{AGR20}.
Most recently \cite{AGKMMR24} have shown that an $n$-qubit unitary matrix $U$ can be exactly represented with the gateset $\calG_{2^k}$ iff $U$'s entries are in $\ZZ[1/2,\zeta_{2^k}]$ (see \cref{thm:AGKMMR24}), where $\zeta_{2^k}=e^{2\pi\iu/2^k}$ is a primitive $2^k$-th root of unity, $\calG_2=\{\Xg,\CXg,\CCXg,\Hg\otimes \Hg\}$ (X, CNOT, Toffoli, Hadamard), $\calG_{4} = \{\Xg,\CXg,\CCXg,\Sg,\zeta_8\Hg\}$, and for $k\ge3$, $\calG_{2^k} = \{\Hg,\CXg, \Tg_{2^k}\}$ with $\Tg_{2^k} = \begin{bsmallmatrix}1&0\\0&\zeta_{2^k}\end{bsmallmatrix}$ ($\Tg \equiv \Tg_8$, $\Sg \equiv \Tg_4$).

Surprisingly, $\QMAo$ has recently appeared in computational topology:
Crichigno and Cade \cite{CC24} proved that the $k$-local cohomology problem is $\QMAo$-hard, for a cohomology inspired by supersymmetric systems.
% There, the connection to $\QMAo$ is more direct, since the cohomology directly acts on the fermionic Fock space.
% The connection of (co)-homology to $\QMAo$ was first pointed out by Cade and Crichigno \cite{CC24}.
Crichigno and Kohler \cite{CK24} showed that the problem of determining whether the clique complex of a given graph has a hole of a given dimension, is $\QMAo$-hard.
This problem was first defined by Kaibel and Pfetch \cite{KP02}.
The clique complex of a graph $G$ is the simplicial complex obtained by declaring every $(k+1)$-clique in $G$ to be a $k$-simplex.
They also showed that the problem is contained in $\QMA$ when the combinatorial Laplacian has an inverse polynomial spectral gap, and conjectured $\QMAo$-hardness.
King and Kohler \cite{KK24} made progress towards this conjecture by showing that a modified version of the clique homology problem on \emph{weighted} graphs with an inverse polynomial spectral gap, is $\QMAo$-hard and contained in $\QMA$.
Note that the weights are only used to lower bound the gap of the Laplacian, but do not affect the cohomology itself.

Although the results of \cite{CK24,KK24} show computing the sizes of cohomology groups exactly is hard in general, the quantum algorithm for topological data analysis by Lloyd, Garnerone and Zanardi \cite{LGZ16} can approximate the number of holes in persistent homology (see also the survey \cite{Was18}).
Schmidhuber and Lloyd \cite{SL23} proved that computing Betti numbers of a clique complex (i.e. the number of $k$-dimensional holes) is $\#\p$-hard, and that the decision clique homology problem is $\NP$-hard.

The clique homology problem is really the first problem where the gateset in the definition of $\QMAo$ truly matters.
Quantum SAT and even the cohomology problem of \cite{CC24} provide enough freedom to embed arbitrary local gates.
This is not the case for the clique homology problem, since the homology is completely determined by the cliques of a graph.
There is no way to directly embed complex phases, or even irrational coefficients such as $\sqrt{2}$.
Hence, the $\QMAo$-hardness proofs of \cite{CK24,KK24} are only able to embed quantum SAT instances with projectors of the form $\ketbrab{\phi}$, where $\ket{\phi}$ is proportional to an integer superposition of computational basis states.
To show that the gapped clique homology problem is also $\QMAo$-complete, we therefore need to find a gateset that is both powerful enough to decide the problem with perfect completeness, but also only consists of rational unitaries, so that it can be embedded into the clique homology.

With that motivation, we set out to develop a more nuanced understanding of the gateset of the $\QMAo$-verifier in this paper.
Along the way, we obtain a variety of novel results, including the $\QMAo$-completeness of the gapped clique homology problem, $\PSPACE$-completeness of the decision clique homology problem, the first $\QMAo$-complete $2$-local Hamiltonian problem (in contrast to Bravyi's $2\hQSAT\in\p$), and the first $\QMA_1(2)$-complete Hamiltonian problem (i.e., $\QMA(2)$ with perfect completeness).

\paragraph*{Results.}
Our first result is that for any finite gateset $\calG$ consisting of unitaries with entries in the $2^k$-th cyclotomic field, $\QMAo^{\calG}$ can be simulated using the finite gateset $\calG_{2^k}$.
We use the gatesets $\calG_{2^k}$ defined above from \cite{AGKMMR24} (see \cref{thm:AGKMMR24}).
Additionally, we show that $\QMAo^{\calG_4} = \QMAo^{\calG_2}$, i.e., $\QQ(\iu)$ gatesets can be simulated with just $\QQ$ gates.

\begin{theorem}[\ref{thm:cyclotomic-gateset}, \ref{thm:qma-g2}]\label{thm:first}
  For any finite gateset $\calG$ in $\QQ(\zeta_{2^k})$ with $k\in\NN$, it holds that $\QMAo^{\calG} \subseteq \QMAo^{\calG_{2^k}}$.
  For any finite gateset $\calG$ in $\QQ(\iu)$, it holds that $\QMAo^{\calG} \subseteq \QMAo^{\calG_2}$.
\end{theorem}
Note that $\QMAo^{\calG_2} = \QMAo^{\calG_4} \subseteq \QMAo^{\calG_8} \subseteq \QMAo^{\calG_{16}} \subseteq \cdots$.
For $\BQPo$ we may almost claim universality since we can simulate all cyclotomic gatesets with $\calG_2$.\footnote{This simulation even works for $k=O(\log n)$, which unfortunately does not suffice to simulate an $n$-qubit quantum Fourier transform.}

\begin{theorem}[\ref{thm:bqp1}]\label{thm:firstbqp}
  For any finite gateset $\calG$ in $\QQ(\zeta_{2^k})$ with $k\in \NN$ it holds that $\BQPo^{\calG} \subseteq \BQPo^{\calG_{2}}$.
\end{theorem}

We exhibit complete problems for the classes $\QMAo^{\calG_{2^k}}$.
$\kQSAT^{\FF}$ denotes the quantum $k$-SAT problem with local Hamiltonians in $\FF^{2^k\times2^k}$.
Notably, we also give the first $2$-local $\QMAo$-complete Hamiltonian problem.
We define the $k$-local Exact Hamiltonian problem ($\kELH$, see \cref{def:kELH}) as the problem of deciding whether a $k$-local Hamiltonian on $n$ qubits has an eigenvalues of exactly $0$ or all eigenvalues are inverse polynomially bounded away from $0$ (i.e., $\sigma_1(H)=0$, or $\sigma_1(H)\ge 1/\poly(n)$ for $\sigma_1$ the smallest singular value).
We have $\kQSAT \subseteq \kELH \subseteq 2k\hLH$, but
$\kELH$ has a weaker promise than $\kQSAT$: In the YES-case only a nonempty nullspace is required, but $H$ may still be frustrated (i.e. the local Hamiltonians do not have a common ground state).
$\kELH$ has a stronger promise than $\kLH$: In the YES-case not just low energy is required, but an eigenvalue of \emph{exactly} $0$.

\begin{theorem}[\ref{thm:4SAT}, \ref{thm:3SAT}, \ref{thm:2LH-complete}]
  $4\hQSAT^{\QQ(\zeta_{2^k})}$ is complete for $\QMAo^{\calG_{2^k}}$ for all $k\in\NN$.\\
  $3\hQSAT^{\QQ(\zeta_{2^k})}$ is complete for $\QMAo^{\calG_{2^k}}$ for all $k\ge 3$.\\
  $2\hELH^{\QQ(\zeta_{2^k})}$ is $\QMAo^{\calG_{2^k}}$-complete for all $k\ge3$.
\end{theorem}

We can generalize \cref{thm:first} to $\QMAo(2)$, and show the first complete Hamiltonian problem $\ESSH$ (Exact Separable Sparse Hamiltonian problem, see \cref{def:ESSH}), which is the problem of deciding whether a sparse Hamiltonian $H$ has a separable nullstate, or for all product states, we have $\norm{H\ket{\psi_1}\ket{\psi_2}}\ge1/\poly$.
We also give the first complete problem, $\AESSH$ (Almost Exact Separable Sparse Hamiltonian), for ``high precision'' $\QMA(2)$ (i.e. with a promise gap of less than $1/\poly$), which is an approximate version of $\ESSH$ where in the YES-case the existence of a product state with negligible $\norm{H\ket{\psi_1}\ket{\psi_2}}$ suffices.
This solves an open problem of \cite{GR23}.

\begin{theorem}[\ref{thm:qma(2)-g2}, \ref{thm:qma(2)-cyclotomic-gateset}]
  For any finite gateset $\calG$ in $\QQ(\zeta_{2^k})$ with $k\in\NN$, it holds that $\QMAo^{\calG}(2) \subseteq \QMAo^{\calG_{2^k}}(2)$.
  For any finite gateset $\calG$ in $\QQ(\iu)$, it holds that $\QMAo^{\calG}(2) \subseteq \QMAo^{\calG_2}(2)$.
\end{theorem}

\begin{theorem}[\ref{thm:QMA(2)}, \ref{thm:AESSH}]
  $\ESSH^{\QQ(\zeta_{2^k})}$ is $\QMAo^{\calG_{2^k}}(2)$-complete for all $k\in\NN$.\\
  $\AESSH_\epsilon$ is complete for $\epsilon\mhyphen\QMA(2)\coloneq\bigcup_{c\in 1-\epsilon^{\omega(1)},s\in1-\epsilon^{O(1)}}\QMA_{c,s}(2)$ with $\epsilon \in n^{-\Omega(1)}$.
\end{theorem}

Lastly, we give the first completeness results for the clique homology problem.
$\CH$ is the problem of determining whether the clique complex of a given graph has a hole of a given dimension.
$\GCH$ additionally promises that in the NO-case, the combinatorial Laplacian has a minimum eigenvalue of $1/\poly(n)$, where $n$ is the number of vertices.
The following result follows from the above theorem and plugging the gateset $\calG_2$ into the machinery of \cite{KK24}\footnote{Note that the $\PSPACE$-completeness result could already be obtained by using the earlier framework of \cite{CK24}, but for simplicity we only developed gadgets for $\calG_2$ in the framework of \cite{KK24}.}, and the fact that $\QMAo$ with exponentially small promise gap equals $\PSPACE$ \cite{Li22}.

\begin{theorem}[\ref{thm:GCH}, \ref{thm:CH}] 
  $\GCH$ is $\QMAo^{\calG_2}$-complete.\\
  $\CH$ is $\PSPACE$-complete.
\end{theorem}

This improves the previous state of the art, which was $\QMAo$-hardness and $\QMA$-containment for $\GCH$ \cite{KK24}, and $\QMAo$-hardness for $\CH$ \cite{CK24}.
$\GCH$ is particularly nice as a $\QMAo$-complete problem since the ``perfect completeness'' is inherent to the problem: there exists a hole iff the kernel of the Laplacian is nonempty.
For $\kQSAT$, there is (almost) no physical difference between $\lmin(H) = 0$ and $\lmin(H) < 1/\exp(n)$, whereas even $\lmin(\Delta) > 1/\exp(n)$ still just implies there is no hole, where $\Delta$ is the combinatorial Laplacian.

\cite{SL23} shows that computing the Betti number of a clique complex is $\mathsf{\#P}$-hard, which is not directly comparable since that is not a decision problem.
However, note that the $k$-th Betti number $\beta_k$ denotes the number of $k$-dimensional holes, and so we prove that only deciding whether $\beta_k\ne0$ is already $\PSPACE$-hard.
On its face, the $\PSPACE$-completeness of $\CH$ is a completely classical result, which we remarkably obtain using techniques from quantum complexity theory.

\paragraph*{Techniques.} \emph{Linear combinations of unitaries (LCU).} Our main technical contribution is as follows.
Suppose we wish to apply some unitary $U$ to our state, but we do not know how to do that with our gateset.
Instead, we can write $U$ as a linear combination of unitaries that we \emph{can} efficiently implement.
We follow the idea of Childs and Wiebe \cite{CW12} in the context of Hamiltonian simulation, which allows us to apply $U$ exactly after a successful measurement.
Conveniently, the Pauli matrices $\{\Ig,\Xg,\Yg,\Zg\}$ form an orthonormal basis, and thus we can write any unitary as $U=\sum_{x}a_x P_x$, where the $P_x$ are tensor products of Paulis.
If we can prepare $\ket{U} = \sum_x a_x\ket{x}$, then we can conditionally apply $P_x$ to an input state $\ket{\psi}$ to obtain $\sum_{x}a_x\ket{x}\otimes P_x\ket{\psi}$. Projecting the first register onto the ``all $\ket{+}$'' state then gives $\sum_{x}a_xP_x\ket{\psi}=U\ket{\psi}$ in the second register.
Note that if $U$ is in $\QQ(\zeta_{2^k})$ ($k\ge2$), then $\ket{U}$ is also in $\QQ(\zeta_{2^k})$.
We show how to efficiently prepare such states with postselection.
In the context of $\QMAo$, postselection means that the verifier accepts if postselection fails to maintain perfect completeness.
Although the success probability of applying $U$ in this way is only $2^{-2n}$ for an $n$-qubit operator, we can boost it exponentially close to $1$ using the \emph{oblivious amplitude amplification} technique of \cite{BCCKS14}.

\emph{Simulating cyclotomic fields.} In order to implement gates $\calG_{2^k}$ with $\calG_{2}$, we extend McKague's technique of simulating complex gates with real gates \cite{McK10,McK13} to cyclotomic gates.
We can write any $a \in \QQ(\zeta_{2^k})$ as $a=\sum_{i=0}^{2^{k-1}-1}a_i\zeta_{2^k}^i$ with all $a_i\in\QQ$.
Then we represent $a$ as a vector $\ket{v(a)} = \sum_{i}a_i\ket{i}$, and multiplication by $b$ in $\QQ(\zeta_{2^k})$ becomes matrix multiplication by some $M_b$ in $\QQ$ so that $M_b\ket{v(a)}=\ket{v(ab)}$.
Formally, there exists a field isomorphism from $\QQ(\zeta_{2^k})$ to a subfield of invertible $\QQ^{2^{k-1}\times 2^{k-1}}$ matrices.
Applying this field isomorphism elementwise to a unitary results in an orthogonal matrix.
Formally, we have a group homomorphism $\Psi:\unitary(N,\QQ(\zeta_{2^k}))\to\rmO(2^{k-1}N,\QQ)$.\footnote{This only works for the $2^k$-th roots of unity, since they have the unique property $M_b^T = M_{\overline{b}}$, which preserves orthogonality when applied to a unitary.}

\emph{Nullspace testing.}
We observe that the LCU technique also applies to non-unitary matrices.
We can directly apply the Hamiltonian to a quantum state (i.e., not $e^{\iu H}$ as in Hamiltonian simulation).
However, the success probability is proportional to $\norm{H\ket{\psi}}$.
Thus, if $\ket{\psi}\in\Null(H)$, then multiplying $\ket\psi$ with $H$ fails with probability $1$.
Hence, we can check in $\QMAo$ whether a frustrated Hamiltonian has a nonempty nullspace.
We further extend this technique to sparse Hamiltonians by applying tools from Hamiltonian simulation \cite{BCCKS14,KL21} (modified for cyclotomic fields), to write a sparse Hamiltonian as a linear combination of $1$-sparse unitaries.

\emph{Hamiltonians.}
We construct our Hamiltonians via the Nullspace Connection Lemma of \cite{RGN24}, which effectively gives a blueprint for circuit-to-Hamiltonian constructions and only requires us to prove properties of local gadgets.
We do this computationally \cite{sup}.
Nevertheless, our $2$-local $\QMAo$-complete Hamiltonian requires some new tricks, since the $2$-local $\QMA$-complete circuit-to-Hamiltonian construction of \cite{KKR05} does not have the history state as eigenstate.
We follow the idea of splitting the computation path inside the history state, conditioned on a computational register \cite{ER08,GN13,RGN24}.
The new idea is to split the computation path conditioned on the eigenstates of a single-qubit gate (not just on $\ket{0},\ket{1}$).
Then we can apply a single-qubit gate as a ``zero-qubit'' gate (i.e. a global phase), which we can implement with a $2$-local transition by using a ``one-hot encoding'' for the clock register.

\paragraph*{Open questions.}
Despite having made significant progress, the big question whether $\QMAo$ equals $\QMA$ remains wide open.
Does there exist a classical oracle separation?
We also have ``smaller'' open questions.
Can \cref{thm:firstbqp} be extended to $\QMAo$, i.e., does a single gateset suffice to simulate all $2^k$-th cyclotomic gatesets for $\QMAo$? The issue here is that our encoding allows a malicious prover to send an effective ``zero-state'', as for $k>2$ the powers of $\zeta_{2^{k}}$ are not linearly independent in the complex plane.
More generally, can we extend these results to other roots of unity, i.e., not just powers of two?

On a more positive note, we have shown that $\QMAo$ can solve other problems than just the frustration-free local Hamiltonian problem (i.e. quantum SAT).
Perhaps our techniques will open the door for further $\QMAo$-completeness results.
A few Hamiltonian completeness results are still open for $k<3$ ($3\hQSAT$, $2\hELH$), where our constructions need the eighths root of unity to implement to implement CNOT and Hadamard gates, respectively.

Finally, the $\QMAo$-completeness of $\GCH$ without vertex weights remains open, 

\paragraph*{Organization.}
In \cref{sec:preliminaries}, we define our problems, complexity classes, and introduce clique homology.
In \cref{sec:gateset}, we show how to simulate arbitrary cyclotomic gates with the $\calG_{2^k}$ gatesets.
In \cref{sec:qsat}, we adapt existing $\QMAo$-completeness results for the Quantum SAT problem to our definitions of $\QMAo$.
In \cref{sec:2local}, we show that even $2$-local Hamiltonian problems can be $\QMAo$-hard.
In \cref{sec:sparse}, we adapt our techniques to sparse Hamiltonian problems, including $\QMAt$ and clique homology.

\section{Preliminaries}\label{sec:preliminaries}

In this section, we introduce notation and give formal definitions of the complexity classes and problems used in this paper.

\subsection{BQP and QMA with perfect completeness}

For a quantum verifier circuit $Q$ with $n_1$ ancilla qubits and $n_2$ proof qubits, we define the acceptance probability on input $\ket{\psi}\in\CC^{2^{n_2}}$ as 
\begin{equation}
  \pacc(Q, \psi) = \Tr\left((\ketbrab{1}_1\otimes I_{2,\dots,n_1+n_2})\,Q\, (\ketbrab{0}^{\otimes n_1}\otimes \ketbrab{\psi})\right).
\end{equation}
For a \BQP circuit, we just write $\pacc(Q)$ for the acceptance probability.
We denote the ancilla register by $\calA$ and the proof register by $\calB$.

\begin{definition}[$\QMAo$]\label{def:QMA1}
  A promise problem $A=(\Ayes,\Ano)$ is in $\QMAo^{\calG}$ if there exists a poly-time uniform family of quantum circuits $\{Q_x\}$\footnote{Here, we allow the ``circuit constructor'' access to the instance $x$, as in \cite{Bra06}. Sometimes $\QMA$ is also defined with a universal family of circuits $\{Q_n\}$, i.e., the ``circuit constructor'' only has access to the size of the instance. These definitions are equivalent for the gatesets we consider, i.e., $\calG_{2^k}$ of \cite{AGKMMR24}.} and polynomials $n_1,n_2$, where each $Q_x$ only uses gates from the gateset $\calG$ and acts on $n_1(\abs{x})$ ancilla qubits and $n_2(\abs{x})$ proof qubits, such that:
  \begin{itemize}
    \item (Completeness) If $x\in\Ayes$, then there exists a proof $\ket{\psi}\in\CC^{2^{n_2}}$ with $\pacc(Q_x,\psi) = 1$.
    \item (Soundness) If $x\in\Ano$, then for all proofs $\ket{\psi}\in\CC^{2^{n_2}}$, $\pacc(Q_x,\psi) \le 1-1/\poly(\abs{x})$.
  \end{itemize}
  We additionally define $\QMAo^{\calG}(2)$ in the same way, but the proof is a product state $\ket{\psi} = \ket{\psi_1}\otimes\ket{\psi_2}$ with $\ket{\psi_1},\ket{\psi_2}\in\CC^{2^{n_2}}$.
  We also write $\QMA_{1,c}$ for soundness $c$, i.e., the verifier accepts with probability $\le c$ in the NO-case.
\end{definition}

\begin{definition}[$\BQPo$]\label{def:BQP1}
  A promise problem $A=(\Ayes,\Ano)$ is in $\BQPo^{\calG}$ if there exists a poly-time uniform family of quantum circuits $\{Q_x\}$ and polynomial $n_1$, where each $Q_x$ only uses gates from the gateset $\calG$ and acts on $n_1(\abs{x})$ ancilla qubits, such that:
  \begin{itemize}
    \item (Completeness) If $x\in\Ayes$, then $\pacc(Q_x) = 1$.
    \item (Soundness) If $x\in\Ano$, then $\pacc(Q_x) \le 1-1/\poly(\abs{x})$.
  \end{itemize}
\end{definition}

Standard error reduction results also hold as long as we have at least the gates
\begin{equation}\label{eq:G0}
  \calG_2 \coloneq  \{\Xg,\CXg,\CCXg,\Hg\otimes\Hg\}\qquad \text{(Pauli $X$, CNOT, Toffoli, Hadamard)},
\end{equation}
since they allow us to implement any qubit unitary with entries in $\ZZ[1/2]$~\cite{AGR20,AGKMMR24}.
This also extends to $\QMAo(2)$ since the product test has perfect completeness \cite{HM13}.

\begin{lemma}
  Let $\calG\supseteq\calG_2$. Then $\BQPo^{\calG},\QMAo^{\calG},\QMAo^{\calG}(2)$ have soundness $1/\exp(n)$.
  Also, $\QMAo^{\calG}(2) = \QMAo^{\calG}(\poly)$.
\end{lemma}

\subsection{Hamiltonian problems}

The $k$-local Hamiltonian problem defined below is considered the canonical $\QMA$-complete problem.

\begin{problem}[$k$-local Hamiltonian problem ($\kLH_\epsilon$)]\label{def:kLH}
  Given a classical description of a $k$-local Hamiltonian $H = \sum_{S\subseteq[n]}^m H_S\otimes I_{[n]\setminus S}$ on $n$ qubits with $\norm{H_S}\le1$, and $\alpha,\beta$ with $\beta-\alpha\ge \epsilon(n)$, decide:
  \begin{itemize}
    \item (YES) $\lmin(H) \le \alpha$.
    \item (NO) $\lmin(H) \ge \beta$.
  \end{itemize}
\end{problem}

\begin{definition}\label{def:HP}
  Let $\HP_\epsilon$ one of the Hamiltonian problems defined in this section.
  We omit $\epsilon$ if $\epsilon \in n^{-O(1)}$, i.e., $\HP \coloneq  \bigcup_{\epsilon \in n^{-O(1)}} \HP_{\epsilon}$.
\end{definition}

The quantum $k$-SAT problem is considered the canonical $\QMAo$-complete problem.
It differs from the local Hamiltonian problem in that the local terms are required to be projectors and in the YES-case they must have a common (zero energy) ground state (i.e. $H$ is frustration-free).
The $\QMAo$-completeness of $\kQSAT$ is quite subtle, however, due to the requirement of perfect completeness. 
Bravyi~\cite{Bra06} defined $\kQSAT$ in the same way as below, i.e., the Hamiltonian is restricted to a subfield $\FF\subseteq\CC$ with exact representation, and shows completeness for $\QMAo^{\calG}$, where $\calG$ contains all three-qubit gates with entries in $\FF$.
Gosset and Nagaj~\cite{GN13} instead consider the gateset $\{\Hg,\CXg,\Tg\}$, and show completeness for $\kQSAT$, where each local projector $\Pi$ is in $\ZZ[\frac1{\sqrt2},\iu]$, or there exists a unitary with only these elements, such that $U\Pi U^\dagger = \ketbraa{(\sqrt{1/3}\ket{000}-\sqrt{2/3}\ket{001})}$.

\begin{problem}[Quantum $k$-SAT ($\kQSAT^{\FF}_\epsilon$)]\label{def:kQSAT}
  Let $\FF\subseteq \CC$ be a field.
  Given a classical description of a $k$-local Hamiltonian $H = \sum_{S\subseteq[n]}^m H_S\otimes I_{[n]\setminus S}$ on $n$ qubits with  $H_S\succeq0$ and $\norm{H_S}\le 1$ \footnote{Commonly the $H_S$ are assumed to be projectors, but here we only require $H_S$ to be positive semidefinite so that, e.g., projectors onto $\ket{+}$ can be represented in $\QQ$. We will see in \cref{thm:4SAT}, that this does not add power for $k\ge4$.
  For the $k=3$ case (see \cref{thm:3SAT}), we leave this issue open for future work.}, decide:
  \begin{itemize}
    \item (YES) $\lmin(H) = 0$.
    \item (NO) $\lmin(H) \ge \epsilon(n)$.
  \end{itemize}
\end{problem}

In this work, we ``bridge the gap'' between the definitions of \cite{Bra06} and \cite{GN13} in that we show completeness for $\kQSAT$ in a \emph{field} for $\QMAo^\calG$ with a \emph{finite} gateset $\calG$.
Our definition of $\kQSAT$ is more general than that of \cite{GN13} and we have $3\hQSAT^{\text{\cite{GN13}}} \subseteq 3\hQSAT^{\QQ(\zeta_8)}$ (see \cref{thm:3SAT}).
Specifically, we give a gateset for every $2^k$-th cyclotomic field $\QQ(\zeta_{2^k})$, where $\zeta_{2^k} = e^{2\pi\iu /2^k}$, i.e., the field extension of $\QQ$ generated by a primitive $2^k$-th root of unity.
Note that the $\Tg$-gate is in $\QQ(\zeta_8)$.

We also define a new variant of the local Hamiltonian problem with a stronger promise in the YES-case, where we only have to distinguish between the cases $\sigma_1(H)=0$ and $\sigma_1(H)\ge 1/\poly(n)$, for $\sigma_1$ the smallest singular value.

\begin{problem}[$k$-local Exact Hamiltonian problem ($\kELH^{\FF}_\epsilon$)]\label{def:kELH}
  Given a classical description of a $k$-local Hamiltonian $H = \sum_{S\subseteq[n]}^m H_S\otimes I_{[n]\setminus S}$ on $n$ qubits such that all entries of the $H_S$ are in $\FF$, decide:
  \begin{itemize}
    \item (YES) $\sigma_1(H) = 0$.
    \item (NO) $\sigma_1(H) \ge \epsilon(n)$.
  \end{itemize}
\end{problem}

Note that $\kELH$ as defined above is equivalent to the problem of deciding whether $H$ has some given eigenvalue $\alpha\in\FF$, or if for all eigenvalues $\lambda$ of $H$, $\abs{\alpha-\lambda}\ge\epsilon(n)$.
The reduction is trivial by transforming $H$ into $H' = H - \alpha I$.
In that sense, we can view $\kELH$ as an ``exact'' variant of the local Hamiltonian problem.

Next, we also define sparse variants of the above problems.
We say an operator $A$ on $n$ qubits is \emph{sparse} if there exists a $\poly(n)$-size circuit that given a row/column of $A$, computes the indices and values of all non-zero entries of that row/column.

\begin{problem}[Sparse Hamiltonian problem ($\SH_\epsilon$)]\label{def:SH}
  Given a classical description of a sparse Hamiltonian $H$ on $n$ qubits and $\alpha,\beta$ with $\beta-\alpha\ge \epsilon(n)$, decide:
  \begin{itemize}
    \item (YES) $\lmin(H) \le \alpha$.
    \item (NO) $\lmin(H) \ge \beta$.
  \end{itemize}
\end{problem}

\begin{problem}[Exact Sparse Hamiltonian problem ($\ESH^{\FF}_\epsilon$)]\label{def:ESH}
  Given a classical description of a sparse Hamiltonian $H$ on $n$ qubits with all entries in $\FF$ and $\norm{H}\le 1$, decide:
  \begin{itemize}
    \item (YES) $\sigma_1(H) = 0$.
    \item (NO) $\sigma_1(H)\ge\epsilon(n)$.
  \end{itemize}
\end{problem}

\begin{problem}[Separable Sparse Hamiltonian problem ($\SSH_\epsilon$) \cite{CS12}]\label{def:SSH}
  Given a classical description of a sparse Hamiltonian $H$ on $2n$ qubits with $\norm{H}\le1$ and $\alpha,\beta$ with $\beta-\alpha\ge \epsilon(n)$, decide:
  \begin{itemize}
    \item (YES) There exists $\ket{\psi}=\ket{\psi_1}\otimes\ket{\psi_2}$ with $\ket{\psi_1},\ket{\psi_2}\in\CC^{2^n}$ such that $\braketb{\psi}{H} \le \alpha$.
    \item (NO) For all such $\ket{\psi}$, $\braketb{\psi}{H} \ge \beta$.
  \end{itemize}
\end{problem}

Note that $\SH\in\QMA$~\cite{CS12} (hardness is trivial) via phase estimation~\cite{NC10} and simulatability of sparse Hamiltonians~\cite{AT03}.
The following separable sparse Hamiltonian problem is the first non-trivial $\QMAo(2)$-complete problem.

\begin{problem}[Exact Separable Sparse Hamiltonian problem ($\ESSH^{\FF}_\epsilon$)]\label{def:ESSH}
  Given a classical description of a sparse Hamiltonian $H$ on $n$ qubits with all entries in $\FF$ and $\norm{H}\le 1$, decide:
  \begin{itemize}
    \item (YES) There exists $\ket{\psi}=\ket{\psi_1}\otimes\ket{\psi_2}$ with $\ket{\psi_1},\ket{\psi_2}\in\CC^{2^n}$ such that $H\ket{\psi}=0$.
    \item (NO) For all such $\ket{\psi}$, $\norm{H\ket{\psi}}\ge\epsilon(n)$.
  \end{itemize}
  We also define an ``almost exact'' version of $\ESSH$, denoted $\AESSH_\epsilon^{\FF}$, where the YES-case only requires $\norm{H\ket{\psi}}\le \epsilon^{\omega(1)}$.
\end{problem}

\begin{definition}[$\FF = \QQ(\zeta_{2^k})$]\label{def:F-zeta}
  For the above problems with $\FF = \QQ(\zeta_{2^k})$ with $k\in\NN$, we require the numerators and denominators of rational coefficients of the Hamiltonian entries to be bounded by $\poly(\epsilon^{-1},n)$ in absolute value.%
  \footnote{For $k\in\{1,2\}$ we still get containment if we allow $\poly(n)$ bit complexity of entries. But for $k\ge3$, we are limited by our integer state preparation routine in \cref{lem:integer-state-cyclotomic}. We allow entries to grow proportionally with $\epsilon^{-1}$, since for $\epsilon\in n^{-\omega(n)}$, state preparation is no longer the bottleneck.}
  For the sparse problems, we also require that the common denominator is bounded by $\poly(\epsilon^{-1},n)$, i.e., $H = H'/h$ with $H'\in\ZZ[\zeta_{2^k}]$ and $h\in\NN,h\le \poly(n)$.
\end{definition}

\subsection{Clique homology}

We give a brief summary of the clique homology problem to the extent that we require for this paper.
We refer the reader to \cite{KK24} for a more complete treatment.
\begin{definition}[Clique complex \cite{CK24,KK24}]\label{def:clique-complex}
  \begin{itemize}
    \item A \emph{simplicial complex} $\calK$ is a collection of subsets $\calK = \calK^0\cup\calK^1\cup\dotsm$, such that $\forall\sigma\in\calK^k\colon \abs{\sigma} = k+1$ and $\forall\sigma\in\calK\;\forall \tau\subset\sigma\colon\tau\in\calK$.
    \item Let $\CK$ be the complex vector space formally spanned by $\calK^k$, picking a conventional ordering for each simplex $\sigma = [v_0,\dots,v_k]$ and identifying $\ket{[\pi(v_0),\dots,\pi(v_k)]} = (-1)^{\sgn(\pi)}\ket{\sigma}$ for all permutations $\pi\in S_{k+1}$. $\sgn(\pi)$ denotes the sign of the permutation and is also known as the \emph{orientation} of the simplex.
    \item For a $k$-simplex $\sigma\in\calK^k$, let $\up(\sigma) \subseteq \calK^0$ be the set of vertices $v$, such that $\sigma\cup\{v\}\in\calK^{k+1}$.
    \item The \emph{coboundary map} is given by $d^k:\CKk{k}\to\CKk{k+1},\;d^k\ket{\sigma} \coloneq  \sum_{v\in\up(\sigma)}\ket{[v]+\sigma}$, where $[v]+\sigma\equiv [v,v_0,\dots,v_k]$ for $\sigma=[v_0,\dots,v_k]$.\footnote{Our notation in the definition of $d^k$ differs from \cite{KK24}, which uses ``set notation'' $\ket{\sigma\cup[v]}$. When the orientation of simplices is important, we use a ``list notation'', i.e., $[v]+\sigma$ to denote prepending $v$ to $\sigma$. This way $d^k\circ d^{k-1}=0$ is easy to see.}
    \item A \emph{chain complex} is a chain of complex vector spaces $\calC^k$ with linear maps $d^k:\calC^k\to\calC^{k+1}$ satisfying $d^k\circ d^{k-1} = 0$ for all $k$.
    \begin{equation*}
      \calC^{-1}\xrightarrow{\;d^{-1}\;}\calC^0\xrightarrow{\;d^{0}\;}\calC^1\xrightarrow{\;d^{1}\;}\calC^2\xrightarrow{\;d^{2}\;}\cdots
    \end{equation*}
    \item The \emph{cohomology groups} are defined as $H^k = {\displaystyle\frac{\Kernel d^k}{\Image d^{k-1}}}$, noting $\Image d^{k-1} \subseteq \Kernel d^k$. The $k$-th Betti number is $\beta_k = \dim H^k$.
    \item Defining an inner product on $\calC^k$ yields a Hilbert space, and we define the \emph{boundary map} as $\partial^k = (d^{k-1})^\dagger$ and the \emph{Laplacian} as $\Delta^k = d^{k-1}\partial^k + \partial^{k+1}d^k$.
    For now we just declare the simplices to be an orthonormal basis, so that $\braket{\sigma}{\tau} = \begin{cases}
      1, &\text{if }\sigma=\tau\\
      0 &\text{otherwise}
    \end{cases}$.
    \item The \emph{homology groups} can then be defined as $H^{k\ast} = {\displaystyle\frac{\Kernel{\partial^k}}{\Image\partial^{k+1}}}$. Here, the homology and cohomology groups are isomorphic (see \cite{CC24}).
    \item To derive a chain complex from a simplicial complex, set $\calC^{-1}(\calK) = \Span\{\ket\emptyset\}\cong \CC$ and define $d^{-1}\ket{\emptyset} = \sum_{v\in\calK^k}\ket{v}$.
    \item The \emph{clique complex} of a graph $G$, denoted $\Cl(G)$ is the simplicial complex consisting of the cliques of $G$.
  \end{itemize}
\end{definition}

\begin{proposition}[\cite{KK24}]
  It holds that $\braketb{\psi}{\Delta^k} = \norm{\partial^k\ket{\psi}}^2 + \norm{d^k\ket{\psi}}$.
  $\Kernel \Delta^k$ is canonically isomorphic to $H^k$, which means that every homology class has a unique \emph{harmonic} representative, i.e., in the kernel of the Laplacian.
\end{proposition}

\begin{problem}[Clique homology problem ($\CH$) \cite{KP02,CC24,CK24}]\label{def:CH}
  Given a graph $G$ and integer $k$, decide whether the $k$-th homology group of $\Cl(G)$ is non-empty:
  \begin{itemize}
    \item (YES) The $k$-th homology group of $\Cl(G)$ is non-trivial: $H^k(G) \ne 0$ (i.e. $\lmin(\Delta^k)=0$).
    \item (NO) The $k$-th homology group of $\Cl(G)$ is trivial: $H^k(G) = 0$ (i.e. $\lmin(\Delta^k)\ne0$).
  \end{itemize}
\end{problem}

\begin{definition}[Weighted chain complex \cite{KK24}]
  A \emph{weighted simplicial complex} is defined by a simplicial complex $\calK$ and a weighting function $w:\calK\to \RR_{\ge0}$, such that $w(\sigma) = \prod_{v\in\sigma}w(v)$ for all $\sigma\in\calK$.
  The inner product on the weighted complex is given by 
  \begin{equation}
  \braket{\sigma}{\tau} = \begin{cases}
      w(\sigma)^2, &\text{if }\sigma=\tau\\
      0 &\text{otherwise}
    \end{cases}.
  \end{equation}
  Then an orthonormal basis for $\CKk{k}$ is $\{\ket{\sigma'}\mid \sigma\in\calK^k\}$ with $\ket{\sigma'} = \frac1{w(\sigma)}\ket{\sigma}$ the unit vector of $\ket{\sigma}$.
\end{definition}

The coboundary and boundary operators act as follows on our new orthonormal basis:
\begin{lemma}[Proof in \cref{sec:proofs}]\label{lem:boundary}
  Let $\sigma=[v_0,\dots,v_k]\in\calK$ in a weighted simplicial complex. Then
  \begin{subequations}\label{eq:boundary sigma'}
  \begin{align}
    d^k\ket{\sigma'} &= \sum_{v\in\up(\sigma)} w(v)\ket{([v]+\sigma)'},\label{eq:boundary sigma':a}\\
    \partial^k\ket{\sigma'} &= \sum_{j=0}^k (-1)^j\cdot w(v_j)\ket{\sigma_{-j}'},\label{eq:boundary sigma':b}
  \end{align}
  \end{subequations}
  where $\sigma_{-j} = [v_0,\dots,v_{j-1},v_{j+1},\dots,v_k]$.
\end{lemma}

Next, we describe the action of the Laplacian $\Delta^k$ on the $\ket{\sigma'}$.
\begin{lemma}[{\cite[Fact 7.4]{KK24}, \cite[Theorem 3.3.4]{Gol02}; proof in \cref{sec:proofs}\protect\footnotemark}]\label{lem:laplacian}
  \footnotetext{We give a proof in the appendix since \cite[Fact 7.4]{KK24} is given without proof and has a very minor bug in that there is a superfluous ``$+1$'' in the first case.}
  Let $\sigma,\tau\in\calK^k$.
  We say $\sigma$ and $\tau$ have a \emph{similar/dissimilar common lower simplex} if there exist $v_\sigma\in\sigma,v_\tau\in\tau$, such that removing $v_\sigma$ from $\sigma$ and $v_\tau$ from $\tau$ gives the same $(k-1)$-simplex $\eta$, and $\ket{\eta}$ has the same/different sign in $\partial^k\ket{\sigma}$ and $\partial^k\ket{\tau}$.
  We say $\sigma$ and $\tau$ are \emph{upper adjacent} if their union forms a $(k+1)$-simplex, i.e., they are both faces of the same $(k+1)$-simplex.
  \begin{equation*}
    \bra{\sigma'}\Delta^k\ket{\tau'} = \begin{cases}
      \sum_{u\in \up(\sigma)} w(u)^2 + \sum_{v\in\sigma}w(v)^2,&\text{if }\sigma=\tau.\\[2mm]
      w(v_\sigma)w(v_\tau),&\parbox{8cm}{if $\sigma$ and $\tau$ have a similar common lower simplex\\ and are not upper adjacent}\\[4mm]
      -w(v_\sigma)w(v_\tau),&\parbox{8cm}{if $\sigma$ and $\tau$ have a dissimilar common lower simplex and are not upper adjacent}\\[2mm]
      0,&\text{otherwise}
    \end{cases}
  \end{equation*}
\end{lemma}
It is easy to see that the Laplacian is sparse and positive semidefinite.

\begin{problem}[Gapped clique homology ($\GCH_\epsilon$) \cite{CK24,KK24}]\label{def:GCH}
  Given a (vertex-weighted) graph $G$ on $n$ vertices and integer $k$, decide whether the $k$-th homology group of $\Cl(G)$ is non-empty with the additional promise that in the NO-case $\lmin(\Delta^k)\ge\epsilon(n)$ for $\Delta^k$ the Laplacian of $\Cl(G)$:
  \begin{itemize}
    \item (YES) The $k$-th homology group of $\Cl(G)$ is non-trivial: $H^k(G) \ne 0$.
    \item (NO) The $k$-th homology group of $\Cl(G)$ is trivial: $H^k(G) = 0$.
  \end{itemize}
  Define $\GCH$ as in \cref{def:HP}.
\end{problem}

\section{(Towards) universal gatesets}\label{sec:gateset}

\cite{AGKMMR24} prove that the gatesets $\calG_{2^k}$ can exactly synthesize all unitaries with entries in $\ZZ[1/2,\zeta_{2^k}]$.
\begin{align}
    \calG_{2} &= \{\Xg,\CXg,\CCXg,\Hg\otimes\Hg\},\\
    \calG_{4} &= \{\Xg,\CXg,\CCXg,\zeta_8\Hg\},\\
    \calG_{2^k} &= \{\Hg,\CXg,\Tg_{2^k}\},\quad \Tg_{2^k} = \begin{pmatrix}
      1&0\\0&\zeta_{2^k}
    \end{pmatrix},
\end{align} 
\begin{theorem}[\cite{AGKMMR24}]\label{thm:AGKMMR24}
  Let $k,m\in\NN$. A $2^m\times2^m$ unitary $U$ can be exactly represented by an $m$-qubit circuit over $\calG_{2^k}$ if and only if $U\in\unitary(2^m,\DD[\zeta_{2^k}])$, where $\DD[\zeta_{2^k}]\equiv \ZZ[1/2,\zeta_{2^k}]$ and $\unitary(N,R)$ denotes the set of unitary $N\times N$ matrices in $R$.
  For $k\le 2$ a single ancilla suffices and $k-2$ ancillas for $k>2$.
\end{theorem}

Note that these gatesets are sufficiently powerful to at least capture $\QCMA$, i.e., it holds for all $k$ that $\QCMA\subseteq\QMAo^{\calG_{2^k}}$ \cite{JKNN12}.
We prove the following results in this section.

\begin{theorem}\label{thm:cyclotomic-gateset}
  For any finite gateset $\calG$ in $\QQ(\zeta_{2^k})$ with $k\ge 2$, it holds that $\QMAo^{\calG} \subseteq \QMAo^{\calG_{2^k}}$.
\end{theorem}

\begin{theorem}\label{thm:qma-g2}
  For any finite gateset $\calG$ in $\QQ(\iu)$, it holds that $\QMAo^{\calG} \subseteq \QMAo^{\calG_2}$.
\end{theorem}

For $\BQPo$ we may almost claim universality since we can simulate all cyclotomic gatesets with $\calG_2$.

\begin{theorem}\label{thm:bqp1}
  For any finite gateset $\calG$ in $\QQ(\zeta_{2^k})$ with $k\in\NN$, it holds that $\BQPo^{\calG} \subseteq \BQPo^{\calG_{2}}$.
\end{theorem}

\subsection{Integer state preparation}

\begin{proposition}[Proof in \cref{sec:proofs}]\label{prop:exact-synthesis-impossible}
  Let $\calG$ be a finite gateset in $\QQ(\zeta_n)$.
  It is impossible to synthesize all rational single-qubit unitaries with $\calG$, even with $\ket{0}$-initialized ancillas.
\end{proposition}

However, we can use postselection to circumvent this issue.
The first technical tool we need is preparing integer states, i.e., states that are proportional to (unnormalized) states with only integer amplitudes.

\begin{lemma}\label{lem:integer-state}
  Let $\ket\psi\propto\sum_{i=0}^{d-1}a_i\ket{i}$ with $a_i\in\ZZ$.
  We can prepare $\ket\psi$ with probability $\ge 1/4d$ with gateset $\calG_2$ in time $\poly(d,\log A)$ and space $O(\log d+\log A)$, where $A \coloneq  \sum_i \abs{a_i}$.
\end{lemma}
\begin{proof}
  For simplicity, it suffices to consider states with all $a_i>0$, as it is trivial to transform such a state into the target state via $\CZg$ gates and permutations.
  Let $e_i \coloneq  \lceil\log a_i\rceil$, and $n\coloneq \lceil\log A\rceil$.
  Prepare $\ket{\psi_0} = H^{\otimes n}\ket{0} \propto \sum_{j=0}^{2^n-1} \ket{j}_\calB$.
  Perform measurement $M = \sum_{j=0}^{A-1}\ketbrab{j}$, which accepts with probability $\ge1/2$.
  The post-measurement state is $\ket{\psi_1} \propto \sum_{j=0}^{A-1} \ket{j}_\calB$.
  Let $k_i = \sum_{j=0}^{i-1}a_j$.
  Using a simple classical circuit, we ``tag'' $\ket{j}$ with $i$ if $j\in[k_i, k_{i+1})$ (i.e. $\ket{0}\ket{j}\mapsto \ket{i}\ket{j}$), obtaining $\ket{\psi_2}\propto \sum_{i=0}^{d-1}\ket{i}_\calA\sum_{j=0}^{a_i-1}\ket{j+k_i}_\calB$, which is further transformed to $\ket{\psi_3}\propto \sum_{i=0}^{d-1}\ket{i}_\calA\sum_{j=0}^{a_i-1}\ket{j}_\calB$ using a classical circuit ($\ket{i}_\calA\ket{j}_\calB\mapsto\ket{i}_\calA\ket{j-k_i}_\calB$).
  Let $a^*\coloneq  \max_i a_i$ and $m = \lceil\log a^*\rceil$, and measure $\calB$ in the Hadamard basis, where $M=\ketbra++^{\otimes m}_\calB$ is the accepting projector.
  For $\ket{\phi} \coloneq  \frac1{\sqrt{a^*}}\sum_{i=0}^{a^*-1}\ket{i}$, we have 
  \begin{equation}
    \bra{+}^{\otimes m}\ket{\phi} = \frac1{\sqrt{2^m \cdot a^*}}\left(\sum_{i=0}^{2^m-1}\bra{i}\right)\left(\sum_{i=0}^{a^*-1}\ket{i}\right) = \frac{a^*}{\sqrt{2^m \cdot a^*}} = \sqrt{\frac{a^*}{2^m}}>\sqrt{\frac12}.
  \end{equation}
  Hence, the measurement accepts with probability at least $1/2d$.
  The post-measurement state is $\ket{\psi_4}\propto M_\calB\ket{\psi_3} \propto \sum_{i=0}^{d-1}a_i\ket{i}_\calA\ket{+}_\calB^{\otimes m}\propto \ket\psi$.

  Note, this algorithm only requires classical circuits on $O(n + \log d)$ qubits as well as (controlled) Hadamard gates, which can be implemented with Toffoli, Hadamard, CNOT, X gates using standard techniques \cite{NC10}.
  We also need to take care to always use two Hadamard gates at the same time, since $\calG_2$ only contains $\Hg\otimes\Hg$.
  So if $n$ is odd, we can apply an additional $\Hg$ to a later unused ancilla during the preparation of $\ket{\psi_0}$.
  If $m$ is odd, we can prepare $\ket{+}\ket{+}$ by applying $\Hg\otimes\Hg$ to two fresh ancillas before the measurement $M$, and then also include one of these ancillas in the measurement.
\end{proof}

\begin{remark}\label{rem:integer-state-boost}
  By repeating the algorithm described in \cref{lem:integer-state} $\poly(d)$ times, we can boost the success probability to $1-1/\exp(d)$.
\end{remark}

\begin{lemma}\label{lem:integer-state-complex}
  Let $\ket\psi\propto\sum_{i=0}^{d-1}(a_i+b_i\iu)\ket{i}$ with $a_i,b_i\in\ZZ$ such that $a_i+b_i\iu \ne 0$ for all $i$.
  We can prepare $\ket\psi$ with probability $\ge 1/16d$ using gateset $\calG_4$ in time $\poly(d,\log A)$ and space $O(\log d+\log A)$, where $A \coloneq  \sum_i \abs{a_i} + \abs{b_i}$.
\end{lemma}
\begin{proof}
  First, prepare state $\ket{\psi_0} \propto \sum_{i=0}^{d-1}(a_i\ket{i} + b_i\iu\ket{i+d})$ using \cref{lem:integer-state}, which succeeds with probability $\ge1/8d$.
  Use a classical circuit to transform to $ \ket{\psi_1}\propto\sum_{i=0}^{d-1} (a_i\ket{i}_\calA\ket{0}_\calB + b_i\iu\ket{i}_\calA\ket{1}_\calB)$.
  Then, measure $\calB$ in Hadamard basis and accept on outcome $\ket{+}$.
  The post-measurement state is then given by $\ket{\psi_2} \propto \sum_{i=0}^{d-1} (a_i+b_i\iu)\ket{i}_\calA\otimes \ket{+}_\calB\propto \ket\psi$.
  The measurement succeeds with probability $1/2$ since on outcome $\ket{-}$, we just get the complex complement of $\ket{\psi}$.
\end{proof}

\cref{lem:integer-state-complex} also trivially extends to the other cyclotomic fields of degree $2$, which are $\QQ(\zeta_3)$ and $\QQ(\zeta_6)$ (see \cite{Rom06} for the necessary background on field extensions and cyclotomic fields).
However, cyclotomic fields of higher degree are more difficult since there it is possible to add integer multiples of powers $\zeta_n$ to get arbitrarily close to $0$.
We can still lower bound the success probability by a polynomial in the sum of absolute values of the integer coefficients:

\begin{lemma}[\cite{Mye86}]\label{lem:small-sum}
  Fix a constant $n\in\NN$. Let $a_0,\dots,a_{n-1}\in \ZZ$, such that there exists an $i$ with $a_i\ne0$.
  Then
  \begin{equation}
    \abs*{\sum_{i=0}^{n-1}a_i\zeta_n^i}\ge \left(\sum_{i=0}^{n-1}\abs{a_i}\right)^{-n}.
  \end{equation}
\end{lemma}

\begin{lemma}\label{lem:integer-state-cyclotomic}
  Let $n=2^k$ and $\ket\psi\propto\sum_{i=0}^{n/2-1}a_i\ket{i}$ with $a_i\in\ZZ[\zeta_n^k]$ with $a_i=\sum_{j=0}^{n/2-1}b_{ij}\zeta_n^j\ne0$ \footnote{It suffices to only use powers of $\zeta_n$ up to the degree of $\zeta_n$, which is $2^{k-1}$ in this case.} for all $i$.
  We can prepare $\ket\psi$ with probability $\ge 1/(d\cdot \poly(A))$ using gateset $\calG_{2^k}$ in time $\poly(d,\log A)$ and space $O(\log d+\log A)$, where $A \coloneq  \sum_{i=0}^{d-1}\sum_{j=0}^{n/2-1} \abs{b_{ij}}$.
\end{lemma}
\begin{proof}
  The procedure is analogous to \cref{lem:integer-state-complex}, but we can only lower bound the success probability of the final measurement by $1/\poly(A)$ via \cref{lem:small-sum}.
\end{proof}

\subsection{Implementing integer gates}

Childs and Wiebe \cite{CW12} show how to implement any non-unitary operation via decomposition into a linear combination of unitaries.
We say a unitary $U$ is integer if there exists an $s^2 \in \ZZ$ such that $sU \in \ZZ[\iu]^{2^n\times 2^n}$.
Here, we use a similar approach to implement any integer unitary with a finite gateset.
It is somewhat similar to the magic state construction by Bravyi and Kitaev \cite{BK05} to implement $Z$-rotations exactly.

\begin{lemma}\label{lem:unitary}
  Let $U$ be an $n$-qubit unitary with $sU\in\ZZ[\iu]$ for $s^2\in\NN$.
  $U$ can be implemented with success probability $2^{-2n} - 2^{-2^{\poly(n)}}$ in time and space $\poly(2^n, \log{s})$ using the gateset $\calG_4$.
\end{lemma}
\begin{proof}
\emph{Applying $U$ probabilistically.}
We begin by describing a procedure to apply $U$ with a success probability of $2^{-2n}$, provided that prior integer state preparation by \cref{lem:integer-state-complex} succeeds.
Decompose $U = \sum_{x\in\bin^{2n}} a_xP_x$ in the Pauli basis with $P_x = \bigotimes_{j=1}^n \sigma_{x_{2j-1}x_{2j}}$, where $\sigma_{00} = \Ig, \sigma_{01} = \Xg, \sigma_{10} = \Yg, \sigma_{11} = \Zg$.
We have $\sum_{x\in\bin^{2n}} \abs{a_x}^2 = 1$, as 
\begin{equation}
  2^n = \fnorm{U}^2 = \Tr UU^\dagger=\sum_{x,y}a_x a_y^\dagger \Tr P_x P_y = 2^n\sum_x \abs{a_x}^2.
\end{equation}
Since $sU$ has a unique representation in the Pauli basis as a vector in $\QQ(\iu)^{2^{2n}}$, we have for all $x$, $sa_x\in\QQ(\iu)$.

To implement $U$, first prepare $\ket{U}_\calA = \sum_{x\in\bin^{2n}} a_x\ket{x}$ using \cref{lem:integer-state-complex}.
Then, conditioned on $\ket{x}_\calA$, apply $(P_x)_\calB$, and measure $\calA$ in the Hadamard basis.
On input $\ket{\phi}_\calB$, we thus get
\begin{equation}
    \ket{U}_\calA\ket{\phi}_\calB = \sum_{x\in\bin^{2n}} a_x\ket{x}_\calA\ket{\phi}_\calB
    \mapsto\sum_{x\in\bin^{2n}} a_x\ket{x}_\calA \otimes P_x\ket{\phi}_\calB
    \eqcolon \ket{\psi'}.
\end{equation}
\newcommand{\yH}{y_H}
\newcommand{\ytilde}{\tilde{y}}
\newcommand{\tsigma}{\tilde{\sigma}}
If we get outcome $y\in\bin^{2n}$ for the Hadamard measurement, we project register $\calA$ onto
\begin{equation}
  \ket{\yH} \coloneq  H^{\otimes 2n}\ket{y} = 2^{-n} \sum_{x\in\bin^{2n}}(-1)^{x\cdot y}\ket{x},
\end{equation}
where $x\cdot y$ denotes the inner product between $x$ and $y$ as vectors.
Hence, the post-measurement state is given by
\begin{equation}
  (\ketbrab{\yH}_\calA\otimes I_\calB)\ket{\psi'}\propto \ket{\yH}_\calA\otimes \sum_{x\in\bin^{2n}} (-1)^{x\cdot y}a_x P_x\ket{\phi}_\calB.
\end{equation}
Note, if $y=0^{2n}$, then $(-1)^{x\cdot y}=1$ for all $x$, and thus we get $U\ket\phi_\calB$.
Otherwise, the following operator is applied: 
\begin{equation}
  U_y = \sum_{x\in\bin^{2n}}(-1)^{x\cdot y}a_xP_x = \sum_{x\in\bin^{2n}}a_x\bigotimes_{i=1}^{n}(-1)^{x_{2i-1}y_{2i-1}+x_{2i}y_{2i}}\sigma_{x_{2i-i}x_{2i}}
\end{equation}
We argue that it is of the form $P_z U P_z$ for some $z\in\bin^{2n}$.
For a fixed $y\in\bin^2$, the individual terms in the tensor product are of the form $\tsigma_x = (-1)^{x\cdot y} \sigma_x$ for all $x\in\bin^{2}$.
\begin{enumerate}[label=(\alph*)]
  \item If $y=00$, we have $\tsigma_x = \sigma_x$.
  \item If $y=01$, the phase $-1$ is injected for $x\in\{01,11\}$ (i.e. $\sigma_x\in\{\Xg,\Zg\}$) and so $\tsigma_x = \Yg\sigma_x\Yg$ as $\Yg \Ig\Yg =  \Ig, \Yg\Xg\Yg = -\Xg, \Yg\Yg\Yg =  \Yg, \Yg\Zg\Yg = -\Zg$.
  \item If $y=10$, the phase $-1$ is injected for $x\in\{10,11\}$ (i.e. $\sigma_x\in\{\Yg,\Zg\}$) and so $\tsigma_x = \Xg\sigma_x\Xg$ as $\Xg \Ig\Xg =  \Ig, \Xg\Xg\Xg =  \Xg, \Xg\Yg\Xg = -\Yg, \Xg\Zg\Xg = -\Zg$.
  \item If $y=11$, the phase $-1$ is injected for $x\in\{01,10\}$ (i.e. $\sigma_x\in\{\Xg,\Yg\}$) and so $\tsigma_x = \Zg\sigma_x\Zg$ as $\Zg\Ig\Zg =  \Ig, \Zg\Xg\Zg = -\Xg, \Zg\Yg\Zg = -\Yg, \Zg\Zg\Zg =  \Zg$.
\end{enumerate}
Hence, $\tsigma_x = \sigma_{\ytilde}\sigma_x\sigma_{\ytilde}$, where $\ytilde = y_2y_1$ (swap the bits of $y$).
Since the correction terms do not depend on $x$, they can be factored out and we have $U_y = P_{\ytilde} U P_{\ytilde}$, where now $\ytilde = y_2y_1y_4y_3\dotsm y_{2n}y_{2n-1}$.

The probability of measuring $y$ is given by 
\begin{equation}\label{eq:unitary-probability}
  \begin{aligned}
  \norm{(\ketbrab{\yH}_\calA\otimes I_\calB)\ket{\psi'}}^2 &=\norm*{\ket{\yH}_\calA\otimes2^{-n}\sum_{x\in\bin^{2n}} (-1)^{x\cdot y}a_x P_x\ket{\phi}_\calB}^2\\
  &= \norm*{2^{-n}U_y\ket{\phi}_\calB}^2  = 2^{-2n},
  \end{aligned}
\end{equation}
as $U_y$ is unitary if $U$ is unitary.

\emph{Complexity analysis.} 
We can compute the Pauli coefficients efficiently with the Hilbert-Schmidt product $a_x = 2^{-n}\Tr(P_x U)$.
By \cref{lem:integer-state-complex}, we can prepare $\ket{U}$ in time $\poly(2^n, \log{s})$ with success probability $p\ge2^{-2n-4}$.
If we run the state preparation $r/p$ times, then the probability that all fail is at most $(1-p)^{r/p}\le e^{-r}$.
\end{proof}

\cref{lem:unitary} also extends to cyclotomic integer unitaries, but with the caveat that the success probability now depends on the unitary as in \cref{lem:integer-state-cyclotomic}.

\begin{lemma}\label{lem:unitaryk3}
  Let $k>2$ and $U$ be an $n$-qubit unitary with $sU\in\ZZ[\zeta_{2^k}]$ for $s^2\in\NN$.
  $U$ can be implemented with success probability $1/\poly(2^{2n},s)$ in time and space $\poly(2^n, \log{s})$ using the gateset $\calG_{2^k}$.
\end{lemma}

\subsection{Probability amplification}

Now if we want to use \cref{lem:unitary} to implement an entire circuit consisting of $m$ gates, then success probability becomes an issue, since the procedure would need to succeed $m$ times in a row, which has probability $2^{-\Omega(m)}$.
Fortunately, we can amplify the success probability of \cref{lem:unitary} exponentially close to $1$, which also brings the overall success probability exponentially close to $1$.

As noted in the proof, even if the final Hadamard measurement fails by giving $y\ne0^{2n}$, we still know that some unitary $U_y$ was applied.\footnote{If $U$ is Clifford, then a Pauli correction suffices and we can succeed with probability $1$ (ignoring state preparation).}
So we can repeat the procedure with $U' = UU_y^\dagger$, which again succeeds with probability $2^{-2n}$.
Unfortunately, the bit complexity of $U'$ grows exponentially in the worst case, and therefore we only have a logarithmic number of attempts before the correction terms grow too big.
Still, if we treat $n$ as a constant, then the success probability is also a constant and $O(\log(m))$ attempts suffice for a success probability of $1-1/\poly(m)$.
Otherwise, the runtime would be doubly exponential in $n$.

The increasing bit complexity prohibits this approach in the $k\ge3$ setting (see \cref{lem:unitaryk3}).
Thus, we instead use the \emph{oblivious amplitude amplification} technique of \cite{BCCKS14}, which we restate below for convenience.

\begin{lemma}[{Oblivious amplitude amplification \cite[Lemma 3.6]{BCCKS14}\protect\footnotemark}]\label{lem:oblivious-amplification}
  \footnotetext{We use the lemma as stated in the updated arXiv version of the paper.}
  Let $U$ and $V$ be unitary matrices on $\mu+n$ qubits and $n$ qubits, respectively, and let $\theta\in(0,\pi/2)$.
  Suppose that for any $n$-qubit state $\ket{\psi}$,
  \begin{equation}\label{eq:oblivious:U}
    U\ket{0^\mu}\ket{\psi} = \sin(\theta)\ket{0^\mu}V\ket{\psi} + \cos(\theta)\ket{\Phi^\perp},
  \end{equation}
  where $\ket{\Phi^\perp}$ is an $(\mu+n)$-qubit state that depends on $\ket{\psi}$ and satisfies $\Pi\ket{\Phi^\perp}=0$, where $\Pi \coloneq \ketbrab{0^\mu}\otimes I$.
  Let $R \coloneq 2\Pi - I$ and $S \coloneq -URU^\dagger R$. Then for any $\ell\in\ZZ$,
  \begin{equation}\label{eq:oblivious:S}
    S^\ell U\ket{0^\mu}\ket{\psi} = \sin\bigl((2\ell+1)\theta\bigr)\ket{0^\mu}V\ket{\psi} + \cos\bigl((2\ell+1)\theta\bigr)\ket{\Phi^\perp}.
  \end{equation}
\end{lemma}

\begin{lemma}\label{lem:simulate-Q(i)}
  Let $\calG$ be a fixed finite gateset in $\QQ(\zeta_{2^k})$ with $k\ge2$.
  We can simulate each gate of $\calG$ using $\calG_{2^k}$ with success probability $1-1/\exp(m)$ in time $\poly(m)$.
\end{lemma}
\begin{proof}
  The idea is to apply the $n$-qubit gate $V\in\calG$ using \cref{lem:unitary} or \cref{lem:unitaryk3}.
  Note that both in integer state preparation (\cref{lem:integer-state,lem:integer-state-complex,lem:integer-state-cyclotomic}), and integer gate application (\cref{lem:unitary,lem:unitaryk3}), all ancillas are measured in such a way that their state is fully determined in the case of success.
  Hence, we can construct a unitary $U$ that implements $U$ as in \cref{eq:oblivious:U} with $\sin(\theta)=\sqrt{p}$ for success probability $p$.
  In order to achieve exponentially small failure probability, we need to adjust the success probability so that $(2\ell+1)\theta\approx \pi/2$.

  Let $c = 2^{-r}c'$ with $c'\in [2^r]$ and even $r$ such that $c\sqrt{p} = \sin((\pi/2)/(2\ell+1)) + O(1/\exp(m))$ for some $\ell,r\in\poly(m)$.
  Note that we can efficiently compute $p$ and $c$ since we are dealing with fixed-size circuits.
  We can build a circuit $U'$ such that $U'\ket{0^{r+1}} = c\ket{0^{r+1}} + \sqrt{1-c^2}\ket{\phi^\perp}$ with $\braket{0^{r+1}}{\phi^\perp}=0$.
  $U'$ first applies $\Hg^{\otimes r}$ to the first $r$ qubits, then flips the last qubit conditioned on the first $r$-qubits representing an integer $\ge c'$, and then applies $\Hg^{\otimes r}$ again. Thus,
  \begin{equation}
    \bra{0^{r+1}} V'\ket{0^{r+1}} = \bra{+}^r \frac{1}{\sqrt{2^r}}\sum_{j=0}^{2^r-1}\ket{j}\braket{0}{j\ge c'} = \frac{1}{2^r} \sum_{i=0}^{2^r-1}\sum_{j=0}^{2^r-1}\braket{i}{j}\braket{0}{j\ge c'} = \frac{c'}{2^r}=c,
  \end{equation}
  where $\ket{j\ge c'}\equiv \ket{1}$ for $j\ge c'$ and $\ket0$ otherwise.
  Hence, we can use $\wtU \coloneq U'\otimes U$ to satisfy \cref{eq:oblivious:U} with $\sin((2\ell+1)\theta) \ge 1-1/\exp(m)$.
  We can implement the unitary $S = -\wtU R\wtU^\dagger R$ from \cref{lem:oblivious-amplification} using $\calG_{2^k}$ since $U$ can be implemented with $\calG_{2^k}$ by \cref{lem:unitary,lem:unitaryk3}, and $U',R$ with $\calG_2$.
  Thus, $S^\ell \wtU$ implements $V$ with success probability $1-1/\exp(m)$ in time $\poly(m)$.
\end{proof}

\begin{proof}[Proof of \cref{thm:cyclotomic-gateset}]
  Follows now directly from \cref{lem:simulate-Q(i)}.
\end{proof}

\begin{remark}\label{rem:nonunitary}
  \cref{lem:unitary} only requires $U$ to be unitary to have success probability $2^{-2n}$ for all outcomes $y$ in \cref{eq:unitary-probability} independent of the input state.
  If $U$ is not unitary, then the success probability on input is given by
  \begin{equation}\label{eq:nonunitary-probability}
    \frac{1}{\fnorm{U}^2}\cdot \norm*{\vphantom{\big|}U_y\ket{\phi}}^2,
  \end{equation}
  where $\ket{\phi}$ is the input.
  Notably, the success probability is $0$ if $U\ket{\phi}=0$.
\end{remark}

\subsection{Simulating cyclotomic gates with integer gates}

In this section, we show how to use above constructions to simulate gatesets with entries from a cyclotomic field $\KK=\QQ(\zeta_{n})$ with $n=2^k$.
The case $k=2$ was already proven by McKague~\cite[Section 2.4]{McK10}, and used prove that $\QMA(k),\QIP(k),\QMIP,\QSZK$ over $\R$ are equivalent to their original definitions over $\CC$ \cite{McK13}.

Let $d=2^{k-1}$ be the degree of $\zeta_n$.
We encode $a=\sum_{i=0}^{d-1}a_i\zeta_n^i\in \KK$ as $\ket{v(a)} = \sum_{i=0}^{d-1}a_i\ket{i}$ with all $a_i\in\QQ$, and a vector $\ket{\psi}\in \KK^{N}$ as 
\begin{equation}\label{eq:encoding}
  \ket{\psi} = \sum_{i=0}^{N-2} a_i\ket{i} \qquad\mapsto\qquad \ket{v(\psi)} \propto \sum_{i=0}^{d-1} \ket{i}_\alpha\ket{v(a_i)}.
\end{equation}
Note that this encoding is unique.

\begin{lemma}\label{lem:cyclotomic-simulation}
  There exists a group homomorphism $\Psi:\unitary(N,\KK)\to\rmO(dN,\QQ)$ such that $\Psi(U)$ implements $U$ inside the encoding, i.e., $\Psi(U)\ket{v(\psi)} = \ket{v(U\ket{\psi})}$ for all $\ket{\psi}\in\KK^N$.
\end{lemma}
\begin{proof}
$\KK$ is isomorphic to $\QQ[x]/\Phi_n(x)$, where $\Phi_n(x)$ is the $n$-th cyclotomic polynomial, which has degree $d=\varphi(n)$, where $\varphi$ is Euler's totient function.
An integral basis of $\KK$ is given by $\{1,\zeta_n,\dots,\zeta_n^{d-1}\}$.
Thus, we can treat $\KK$ as a $d$-dimensional vector space over $\QQ$.
Any $a\in \KK$ has therefore a unique decomposition $a = \sum_{i=0}^{n-1}a_i \zeta_n^{i}$ with $a_i\in\QQ$.
Denote by $v_a$ the vector with coefficients $a_0,\dots,a_{d-1}$.
The product of $a,b\in \KK$ is then $ab = \sum_{i,j=0}^{d-1} a_ib_j\zeta_n^{i+j}$.
If we fix $a$, this expression is linear in $v_b$.

Thus, there exists a unique matrix $M_a$ such that $M_av_b = v_{ab}$ for all $a,b$.
In fact, the map $\Psi:a \mapsto M_a$ preserves multiplication and addition, and therefore $\Psi$ is a field isomorphism between $\KK$ and a subfield of $\QQ^{d\times d}$.\footnote{$M_a$ is also known as the regular representation (see e.g. \cite{Die09}). We took the idea of treating the regular representation as a field isomorphism from \cite{Vau15}.}
Note that $n=2^k$, $d=2^{k-1}$ and $\zeta_n^{d} = -1$.
Therefore, we get $M_a^T = M_{\overline{a}}$, because $\overline{\zeta_n} = \zeta_n^{-1} = \zeta_n^{n-1} = -\zeta_n^{d-1}$ with
\begin{equation}
  M_{\zeta_n} = \begin{pmatrix}
    0 &        &   & -1 \\
    1 & \ddots \\
      & \ddots & \ddots\\
      &        & 1 & 0\\
  \end{pmatrix},
  \qquad
  M_{-\zeta_{n}^{d-1}} = -M_{\zeta_n^{d-1}} = -M_{\zeta_{n}}^{d-1} = \begin{pmatrix}
    0 & 1     &   &  \\
      & \ddots & \ddots \\
      &        & \ddots & 1\\
    -1 &        & & 0\\
  \end{pmatrix}.
\end{equation}
Since $\Psi$ is an isomorphism, 
\begin{equation}
  (M_{\zeta_n^j})^T = (M_{\zeta_n}^j)^T = (M_{\zeta_n}^T)^j = M_{\overline{\zeta_n}}^j = M_{\overline{\zeta_n^j}}.
\end{equation}
By linearity, we have
\begin{equation}
  M_a^T = \sum_{i=0}^{d-1} a_iM_{\zeta_n^i}^T = \sum_{i=0}^{d-1} a_iM_{\overline{\zeta_n^i}} = M_{\sum_{i=0}^{d-1}a_i\overline{\zeta_n^i}} = M_{\overline{a}}.
\end{equation}

Thus, applying $\Psi$ to the entries of a unitary gives an orthogonal matrix, i.e., $\Psi\colon \rmU(N, \KK) \to \rmO(dN, \QQ)$ is a group homomorphism:
For $U\in\rmU(N, \KK)$, $\Psi(U)$ is a $N\times N$ block matrix with blocks $\Psi(u_{ij})$ of size $d\times d$.
Let $A = \Psi(U)^T\Psi(U)$ with blocks 
\begin{equation}
  A_{ij} = \sum_{k=1}^N \Psi(U)^T_{ik}\Psi(U)_{kj} = \sum_{k=1}^N \Psi(u_{ki})^T \Psi(u_{kj}) = \Psi\left(\sum_{k=1}^N \overline{u_{ki}}u_{kj}\right) = \Psi(\delta_{ij}) = \delta_{ij}I,
\end{equation} where $\delta_{ij}$ is the Kronecker delta.
Thus, $A = I$ and $\Psi(U)$ is orthogonal.

Writing $U = \sum_{i,j}u_{ij}\ketbra ij$, we get $\Psi(U) = \sum_{i,j} \ketbra ij_\alpha\otimes\Psi(u_{ij})_\zeta$, where the $\zeta$ register contains the corresponding powers of $\zeta_n$, and the $\alpha$ register the actual qubits.
Using the encoding of \cref{eq:encoding}, a quantum state $\ket\psi = \sum_{i=0}^{N-1} a_i\ket{i}$ with $a_i = \sum_{j=0}^{d-1}a_{ij}\zeta_n^j$ is written as \begin{equation}
  \ket{v(\psi)} = \sum_{i=0}^{N-1} \ket{i}_\alpha\ket{v({a_i})}_\zeta = \sum_{i=0}^{N-1}\sum_{j=0}^{d-1} a_{ij}\ket{i}_\alpha\ket{j}_\zeta.
\end{equation}
Thus, 
\begin{equation}
  \begin{aligned}
    \Psi(U)\ket{v(\psi)} &= \sum_{i=0}^{N-1}\sum_{j=0}^{N-1}\ket{i}_\alpha\otimes \Psi(u_{ij})\ket{v(a_j)}_\zeta = \sum_{i=0}^{N-1}\sum_{j=0}^{N-1}\ket{i}_\alpha\otimes \ket{v(u_{ij}a_j)}_\zeta \\
    &= \sum_{i=0}^{N-1}\ket{i}_\alpha\otimes\ket*{v\textstyle\left(\sum_{j=0}^{N-1} u_{ij}a_j\right)}_\zeta = \ket*{v\bigl(U\ket\psi\bigr)}.
  \end{aligned}
\end{equation}
\end{proof}

See \cref{sec:example-cyclotomic} for an example of applying $\Psi$ to $\sqrt{\Hg}$ and the $3$-qubit quantum Fourier transform $F_8$ in $\QQ(\zeta_8)$.

\begin{proof}[Proof of \cref{thm:bqp1}]
  First, we transform the gates of $\calG$ to $\calG'$ in $\QQ$ with \cref{lem:cyclotomic-simulation}.
  Then we simulate these with $\calG_4$ using \cref{lem:simulate-Q(i)}.
  Finally, we apply \cref{lem:cyclotomic-simulation} again to $\calG_4$ to obtain a gateset $\calG_4'$ in $\ZZ[1/2]$, which we can implement exactly with $\calG_2$ by \cref{thm:AGKMMR24}.

  It remains to argue that applying \cref{lem:cyclotomic-simulation} to a $\BQP_1$ verifier yields a valid $\BQP_1$ verifier.
  To avoid clutter, write $\ket{\whpsi} \coloneq  \ket{v(\psi)}, \whU \coloneq  \Psi(U)$.
  Let $V = U_T\dotsm U_1$ be a $\BQP_1$ verifier with soundness $\epsilon$.
  We will show that $\whV = \whU_T\dotsm\whU_1$ is also a valid  $\BQP_1$ verifier for the same problem.
  First, consider the initial state for $V$: $\ket{\psiinit} = \ket{0}_\alpha\ket{x}_\beta$.
  Since $\ket{\psiinit}$ is rational, we have $\ket{\whpsiinit} = \ket{\psiinit}_{\alpha\beta}\ket{0}_\zeta$.

  In the YES-case, $V\ket{\psiinit} = \ket{1}_{\alpha_1}\ket{\phi_1}_{\overline{\alpha_1}} = \ket\phi$.
  Therefore, $\whV\ket{\whpsiinit} = \ket{1}_{\alpha_1}\ket{\whphi_1}_{\overline{\alpha_1}} = \ket\whphi$.
  Thus, perfect completeness is preserved.

  In the NO-case, let $\ket\phi = V\ket{\psiinit} = \ket0\ket{\phi_0} + \ket1\ket{\phi_1}$, where $\ket{\phi_0},\ket{\phi_1}$ are not normalized and $\norm{\ket{\phi_0}}^2\ge 1-\epsilon$ is the rejectance probability of $V$.
  Now consider $\ket{\whphi} = \whV\ket{\whpsiinitp{i}} = \ket{0}\ket{\whphi_0} + \ket{1}\ket{\whphi_1}$, 
  where $p = \braketc{\whphi_0}$ is the rejectance probability of $\whV$.
  Let $\Lambda = \sum_{i=0}^{d-1} I\otimes \zeta_n^i \bra{i}_\zeta$ be the operator that maps $\ket{\whpsi}$ to $\ket{\psi}$.
  Then $\Lambda \ket{\whphi} = \ket{\phi}$ and thus $\Lambda \ket{\whphi_0} = \ket{\phi_0}$.
  Thus, $\norm{\Lambda\ket{\whphi_0}}^2=\norm{\ket{\phi_0}}^2 \ge 1-\epsilon$.
  Hence, $\norm{\ket{\whphi_0}}^2 \ge (1-\epsilon)\norm{\Lambda}^{-2}\ge 1/2n^2$ for small $\epsilon$.

  Thus, the soundness of the $\BQP_1$ verifier $\whV$ is inverse polynomial in $n$.
\end{proof}

For the special case $k=2$, the above idea also works for $\QMA_1$ as shown in \cite{McK10,McK13}, but for $k>2$ there does not seem be a way to prevent the prover from sending a state $\ket{\psi}$ with $\Lambda\ket{\psi}=0$.

\begin{proof}[Proof of \cref{thm:qma-g2}]
  Note that applying \cref{lem:cyclotomic-simulation} to a $\QMAo$-verifier $V$ creates a verifier $V'$ that ``expects'' a proof with real amplitudes.
  With standard error reduction techniques, we can transform $V'$ to $V''$ with soundness $\epsilon$ (on proofs with real amplitudes).
  Then $V''$ still has soundness $4\epsilon$ against proofs with complex amplitudes.
  Hence, we can simulate $\calG$ in the same way as in \cref{thm:bqp1}, but with additional error reduction steps.
\end{proof}

\section{Quantum SAT}\label{sec:qsat}

The first key idea is that we can verify the $k$-local Exact Hamiltonian problem in $\QMAo$ using the insight of \cref{rem:nonunitary}.

\newcommand{\pprep}{p_\mathrm{prep}}
\begin{lemma}\label{lem:kELH}
  $\lELH^{\QQ(\iu)} \in \QMA_{1}^{\calG_{2}}$ and $\lELH_\epsilon^{\QQ(\zeta_{2^k})} \in \QMA_{1,1-\epsilon'}^{\calG_{2^k}}$ with $l\in O(\log n)$,\footnote{For $l\in\omega(1)$, we assume that there are only $\poly(n)$ nonzero $H_S$ terms in $H=\sum_{S}H_S$.} $\epsilon\le n^{-O(1)}$ and $\epsilon'\in\epsilon^{O(1)}$, where $n$ is the number of qubits of the instance.
\end{lemma}
\begin{proof} 
  The first case $\lELH^{\QQ(\iu)} \in \QMAo^{\calG_{2}}$ follows from $\lELH^{\QQ(\iu)} \in \QMAo^{\calG_{4}}$ and \cref{thm:qma-g2}.
  Note that this requires a gap of at least $1/\poly(n)$, since \cref{thm:qma-g2} makes use of error reduction.

  We next describe the $\QMAo$-verifier.
  Let $H= \sum_{S\in\binom{[n]}{l}}H_S$ be the input Hamiltonian with all $\norm{H_S}\le1$.
  We need to distinguish between $\sigma_1(H) = 0$ and $\sigma_1(H)\ge \epsilon$.
  The $\QMA_1$-verifier classically computes the Pauli decomposition of the local $H_S$ terms $H_S = \sum_{x\in\bin^{2l}}a_{S,x}P_x$, where $a_{S,x}\in\QQ(\zeta_{2^k})$.
  Let $2^s \le n^{O(1)}$ be an upper bound on the number of nonzero $A_S$ and identify each such set $S = S_z$ for $z\in\bin^s$.
  Then it prepares the state 
  \begin{equation}
    \ket{H}_\A = \sum_{z\in\bin^{s}}\sum_{x\in\bin^{2l}}a_{zx} \ket{z}_{\A_1}\ket{x}_{\A_2},
  \end{equation}
  using \cref{lem:integer-state-cyclotomic}, which succeeds with probability $\pprep \ge \epsilon^{O(1)}$ since entries are polynomially bounded (see \cref{def:F-zeta}), and accepts if the preparation fails.
  As $\ket{H}$ is not necessarily normalized, we in fact prepare $\eta\ket{H}$, for some $\eta \ge (\sum_z \fnorm{H_{S_z}})^{-2}\ge n^{-O(1)}$.

  Let $U$ be the unitary that applies $(P_x)_{S_z}$ (on register $\B$) conditioned on $\ket{z}_{\A_1}\ket{x}_{\A_2}$.
  $U$ can be efficiently implemented with a polynomial number of $\calG_4$ gates.
  Apply $U$ to $\ket{H}$ and the proof $\ket{\psi}$:
  \begin{equation}
    \ket\phi \coloneq  U\eta\ket{H}_\A\ket{\psi}_\B \coloneq  \sum_{z\in\bin^s}\sum_{x\in\bin^{2l}} \eta a_{zx}\ket{z}_{\A_1}\ket{x}_{\A_2}\otimes ((P_x)_{S_z}\otimes I_{[n]\setminus S_z})\ket{\psi}_\B
  \end{equation}
  Then measure the $\A$ register in the Hadamard basis and reject on outcome $\ket{0_H} \coloneq  (H\ket{0})^{\otimes(s+2l)} = 2^{-s/2-l}\sum_{z\in\bin^s}\sum_{z\in\bin^{2l}}\ket{z}_{\A_1}\ket{x}_{\A_2}$.
  That means, we \emph{reject if we succeed} to apply the non-unitary operator $H$ to the proof, which can only happen if the proof is not in the nullspace of $H$.

  The rejecting projector is then $\Pirej = \ketbrab{0_H}_\A\otimes I_\B$ and the rejectance probability is $\norm{\Pirej\ket{\phi}}^2$ with
  \begin{subequations}
    \begin{align}
      \Pirej\ket{\phi} &= \ket{0_H}_\A \otimes 2^{-s/2-l}\sum_{z\in\bin^s}\sum_{x\in\bin^{2l}}\eta a_{zx}((P_x)_{S_z}\otimes I_{[n]\setminus S_z})\ket{\psi}_\B\\
      &= \ket{0_H}_\A \otimes 2^{-s/2-l}\sum_{z\in\bin^s}\eta(H_{S_z}\otimes I_{[n]\setminus S_z})\ket{\psi}_\B\\
      &= \ket{0_H}_\A \otimes 2^{-s/2-l}\eta H\ket{\psi}_\B.
    \end{align}  
  \end{subequations}
  Thus, the rejectance probability is given by $\prej = 2^{-s-2l}\eta \norm{H\ket\psi}^2$.
  Hence, we get perfect completeness in the YES-case as an honest prover can send a state in the kernel of $H$.
  In the NO-case, we have $\norm{H\ket\psi}\ge \sigma_1(H)\ge\epsilon$, and thus $\prej\ge 2^{-s-2l}\eta\epsilon^2\ge \epsilon^{O(1)}$.
  With state preparation, the verifier rejects with probability $\pprep\cdot \prej \ge \epsilon^{O(1)}$.
\end{proof}

\begin{theorem}\label{thm:4SAT}
  $4\hQSAT^{\QQ(\zeta_{2^k})}$ is complete for $\QMAo^{\calG_{2^k}}$ for all $k\in\NN$.
\end{theorem}
\begin{proof}
  Hardness for $k\ge2$ follows directly by applying Bravyi's $4\hQSAT$ construction~\cite{Bra06}.
  For $\calG_4$, we need to replace the Toffoli gate with a controlled-$\Sg$ gate so that all gates are $2$-local (see \cref{sec:toffoli-from-cs}).
  For $k=1$, we need to slightly alter Bravyi's construction to implement the Toffoli gate $4$-locally.
  The basic idea is to implement logical qu-$5$-its on $4$ physical qubits, and then use the ``triangle Hamiltonian'' gadget of \cite{ER08} to conditionally split the computation paths.
  See \cref{sec:4SAT-G2} for more details.
  Containment follows from \cref{lem:kELH}.
\end{proof}

\begin{theorem}\label{thm:3SAT}
  $3\hQSAT^{\QQ(\zeta_{2^k})}$ is complete for $\QMAo^{\calG_{2^k}}$ for all $k\ge 3$.
\end{theorem}
\begin{proof}
  The $3\hQSAT$ construction of \cite{GN13} is in $\QQ(\zeta_8)$ and uses the gateset $\calG_8$.
  Note that one of their projectors is of the form $\Pi=\frac1{\sqrt3}H$ with $H$ in $\QQ(\zeta_8)$, so we take $H=\Pi/\sqrt{3}$, which is allowed in our definition of $\QSAT$ (see \cref{def:kQSAT}).
  Their construction easily generalizes when replacing $\Tg\equiv\Tg_8$ with $\Tg_{2^k}$.
\end{proof}

\section{A 2-local \texorpdfstring{QMA\textsubscript{1}}{QMA\_1}-complete Hamiltonian problem}\label{sec:2local}

Here we prove the first $2$-local $\QMAo$-complete Hamiltonian problem.

\begin{theorem}\label{thm:2LH-complete}
  $2\hELH^{\QQ(\zeta_{2^k})}$ is $\QMAo^{\calG_{2^k}}$-complete for all $k\ge3$.
\end{theorem}

Bravyi proved that quantum $2$-SAT is in $\p$.
In fact, a frustration-free Hamiltonian always has a ground state that is the tensor product of one- and two qubit states \cite{CCDJZ11,JWZ11}.
Therefore, we will need to use a frustrated Hamiltonian.
Unfortunately, we cannot just use the $2$-local Hamiltonian of \cite{KKR05} since the accepting history state is not an eigenvector.
We need a new $2$-local Hamiltonian construction with the added feature that the history state is also an eigenstate.

Let $V=U_T\dotsm U_1$ be the $\QMAo$-verifier circuit to embed.
Suppose $V$ acts on a register $\calA$ of $n_1$ qubits for the ancillas, and $n_2$ qubits for the proof ($n\coloneq n_1+n_2\le T$).
We embed $V$ into a Hamiltonian $H$ on computational register $\calA$ and clock register $\calC$ of $T+1$ qubits.
$H$ has a similar structure to Kitaev's circuit-to-Hamiltonian construction \cite{KSV02}: $\Hclock$ enforces a logical clock state, $\Hprop$ enforces application of the gates between clock states, $\Hin$ enforces correct initialization of the ancillas, and $\Hout$ enforces acceptance.
\begin{equation}
  H = \Hprop + \Hin + \Hout + \Jclock\Hclock
\end{equation}
The factor $\Jclock$ will be chosen sufficiently large to apply the Projection Lemma~\cite{KKR05} (see \cref{lem:proj}).
$\Hprop$, $\Hin$, $\Hout$ will be defined later.

\subsection{Clock Hamiltonian}
We begin by defining $\Hclock$:
\begin{equation}
 \Hclock = 4T\sum_{1\le i<j\le T+1} \ketbrab{11}_{\calC_i,\calC_{j}} + \sum_{t=1}^{T+1} \Bigl(\ketbrab{0} - T\ketbrab{1}\Bigr)_{\calC_t}
\end{equation}

\begin{claim}\label{claim:Hclock}
  It holds that $\Hclock\succeq0$,
\begin{equation}
  \Null(\Hclock) = \Sclock \coloneq  \Span\left\{\ket{\wht}\mid t \in\{0,\dots,T\}\right\}, \qquad \ket\wht \coloneq  \ket*{0^t\,1\,0^{T-t}},
\end{equation}
and $\gamma(\Hclock)\ge T$, where $\gamma(\cdot)$ denotes the smallest non-zero eigenvalue.
\end{claim}
\begin{proof}
  Since $\Hclock$ is diagonal, its eigenbasis is the computational basis.
  Clearly, $\Hclock\ket{\wht}=0$ for all $t\in\{0,\dots,T+1\}$.
  Further, $\Hclock\ket{0^{T+1}} = (T+1)\ket{0^{T+1}}$, and for $x\in\{0,1\}^{T+1}$ with Hamming weight $h\ge2$, we have $\bra{x}\Hclock\ket{x} = (2(h-1)h-h)T\ge hT$.
\end{proof}

We use a ``one-hot encoding'' for the clock, whereas more commonly a unary clock is used (e.g., \cite{KSV02, KKR05}; see also \cite{CLN18} for an overview of different clock constructions).
Note that a one-hot clock is not possible with quantum SAT, as, if $\Sclock$ is in the kernel of a $k$-local QSAT instance with $k\le T$, then the all-zero state $\ket{0^{T+1}}$ will also be in the kernel.
Next, we leverage the Projection Lemma, reproduced below for convenience, to argue that it suffices to consider $H$ inside the clock space $\Sclock$.

\begin{lemma}[Projection Lemma \cite{KKR05}]\label{lem:proj}
  Let $H=H_1+H_2$ be the sum of two Hamiltonians on Hilbert space $\calH = \calS+\calS^\perp$, such that $\calS$ is a zero eigenspace of $H_2$ and the eigenvectors in $\calS^\perp$ have eigenvalue $J>2\norm{H_2}$. Then
  \begin{equation}
    \lmin(H_1|_\calS) - \frac{\norm{H_1}^2}{J-2\norm{H_1}} \le \lmin(H) \le \lmin(H_1|_\calS),
  \end{equation}
  where $H|_\calS = \Pi_\calS H\Pi_\calS$ and $\Pi_{\calS}$ denotes the projector onto $\calS$.
\end{lemma}

\begin{claim}\label{claim:projection-clock}
  $\lmin(H) \ge \lmin(H\vert_\Sclock) - \epsilon$ for sufficiently large $\Jclock\in (n\epsilon)^{-O(1)}$.
\end{claim}
\begin{proof}
  Let $H = H_1+H_2$ with $H_1 = \Hprop + \Hin + \Hout$ and $H_2 = \Jclock\Hclock$.
  The statement then follows from \cref{lem:proj}, \cref{claim:Hclock}, and $\norm{\Hprop+\Hin+\Hout}\in n^{O(1)}$, which we will see later.
\end{proof}

\subsection{Gate gadgets}

Next, we define $\Hprop$, i.e., the terms of the Hamiltonian that enforce application of the gates between timesteps, so that 
\begin{align}
  \Null(\Hprop + \Hclock) &= \Span\bigl\{\ket{\psihist(x)}\mid x\in\{0,1\}^n\bigr\}\label{eq:NullHclockHprop}\\
  \ket{\psihist(x)} &\coloneq  \frac{1}{\sqrt{T+1}}\sum_{t=0}^T U_{t}\dotsm U_1\ket{x}_{\calA\calB}\ket{\wht}_\calC
\end{align}
To prove \cref{eq:NullHclockHprop}, we will use the ``Nullspace Connection Lemma'' of \cite{RGN24}, which allows us to analyze the ``gate gadgets'' individually, which we restate below for convenience.

\begin{lemma}[Nullspace Connection Lemma \cite{RGN24}]\label{lem:connect}
  Let
  \begin{enumerate}[label=(\arabic*)]
    \item $K_1,\dots,K_m$ be a disjoint partition of the clock states with $u_i,v_i\in K_i,u_i\ne v_i$ for all $i\in[m]$.
    \item $H_{1} = \sum_{i=1}^m H_{1,i}$ be a Hamiltonian such that for all $i\in[m]$:
    \begin{enumerate}
      \item $\Null(H_{1,i}|_{\mathcal{K}_i}) = \Span\{\ket{\psi_i(\alpha_j)}\mid j\in[d]\}$, where $\mathcal{K}_i=\CC^{d}_\calA\otimes\Span\{\ket{v}_\calC\mid v\in K_i\}$, and $\ket{\alpha_1},\dots,\ket{\alpha_d}$ is an orthonormal basis of the ancilla space,
      \item $\exists$ linear map $L_i$ with $L_i\ket{\alpha} = \ket{\psi_i(\alpha)}$ and $L_i^\dagger L_i = \lambda_i I$ for some constant $\lambda_i$,
      \item $H_i$ has support only on clock states $K_i$,
      \item $\norm{\ket{\psi_i(\alpha)}}^2 \eqcolon  \delta_i \in [1, \Delta]$,
      \item $(I_\calA\otimes\bra{u_i}_\calC)\ket{\psi_i(\alpha)} = \ket{\alpha}_\calA$,
      \item $(I_\calA\otimes\bra{v_i}_\calC)\ket{\psi_i(\alpha)} = U_i\ket{\alpha}_\calA$ for some unitary $U_i$.
    \end{enumerate}
    \item $H_2 = \sum_{i=1}^{m-1} h_{v_i,u_{i+1}}(V_i)$ with $h_{v_i,u_{i+1}}(V_i) = I\otimes\ketbrab{v_i} + I\otimes\ketbrab{u_{i+1}} - V_i^\dagger \otimes \ketbra{v_i}{u_{i+1}} - V_i \otimes \ketbra{u_{i+1}}{v_i}$ for unitaries $V_i$.
    \item $\ket{\alpha_{ij}} = V_{i-1}U_{i-1}\dotsm V_1U_1\ket{\alpha_j}$.
  \end{enumerate}
  Then for $H=H_1+H_2$, $\Null(H) = \Span\{\sum_{i=1}^m \ket{\psi_i(\alpha_{ij})}\mid j\in[d]\}$
  and $\gamma(H) = \Omega(\gamma(H_1)/(m^2\Delta))$.
\end{lemma}

The idea is to apply \cref{lem:connect} to $\Hprop|_\Sclock$ and split $\Hprop|_\Sclock = H_1+H_2$ with $H_1=\sum_{t=1}^TH_{1,t}$ implementing the individual gates, and $H_2$ implementing identity transitions between these gadgets.
$H_2$ is straightforward to implement $2$-locally with $V_i=\Ig$ as 
\begin{equation}\label{eq:hij}
  h_{i,j} = \ketbraa{(\ket{10}-\ket{01})}_{\calC_i\calC_j} = \ketbrab{10} + \ketbrab{01} - \ketbra{10}{01}-\ketbra{01}{10}
\end{equation}
with $h_{i,j}|_\Sclock = \ketbraa{(\ket{\wh{i}}-\ket{\wh{j}})}$.

\subsection{Split gadget}

At the heart of the gate gadgets is the ``split gadget''.
It effectively splits the computation path to implement controlled unitaries.
This idea first appeared in the $\QMAo$-completeness proof of $(3,5)\hQSAT$~\cite{ER08}, dubbed the ``triangle Hamiltonian construction'', and similar ideas were also used in \cite{GN13,RGN24}.
Let $\ket{\psi_0},\ket{\psi_1}\in\CC^2$ be orthonormal.
The split gadget $\Hsplit$ acts on a computational register $\calA$ of one qubit and a clock register $\calC$ of $3$ qubits ($\calC_0,\calC_1,\calC_2$).
\begin{equation}
  \begin{aligned}
  \Hsplit =\; &\ketbrab{\psi_1}_{\calA}\otimes\ketbrab{1}_{\calC_{1}} + \ketbrab{\psi_0}_{\calA}\otimes\ketbrab{1}_{\calC_{2}}\,+ \\
  &\ketbra11_{\calC_0} - \ketbra{10}{01}_{\calC_0\calC_1} - \ketbra{10}{01}_{\calC_0\calC_2}\,+ \\
  &\ketbra11_{\calC_1} - \ketbra{10}{01}_{\calC_1\calC_0} + \ketbra{10}{01}_{\calC_1\calC_2}\,+ \\
  &\ketbra11_{\calC_2} - \ketbra{10}{01}_{\calC_2\calC_0} + \ketbra{10}{01}_{\calC_2\calC_1} \\
  \end{aligned}
\end{equation}
\begin{claim}\label{claim:Hsplit}
  $\Hsplit$ is Hermitian with $\Hsplit|_\Sclock\succeq0$, $\Null(\Hsplit|_\Sclock) = \{ \ket{\phi_0},\ket{\phi_1} \}$, where
  \begin{equation}
    \begin{aligned}
      \ket{\phi_0}&=\frac{1}{\sqrt2}\ket{\psi_0}(\ket{100} + \ket{010}),\\
      \ket{\phi_1}&=\frac{1}{\sqrt2}\ket{\psi_1}(\ket{100}+\ket{001}),
    \end{aligned}
  \end{equation} 
  and $\Hsplit\ket{\phi_0} = \Hsplit\ket{\phi_1} = 0$.
\end{claim}
\begin{proof}
  Observe that $\Pi_{\Sclock}\Hsplit\Pi_{\Sclock}=\Hsplit\Pi_{\Sclock}$ since $\Hsplit$ maps logical clock states to logical states, where here $\Pi_{\Sclock}$ acts as identity on $\calA$ and projects $\calC$ onto the clock space $\Sclock$.
  Therefore, it suffices to compute the nullspace of $\Hsplit|_\Sclock$.
  After a change of basis, we have $\ket{\psi_0}=\ket{0},\ket{\psi_1}=\ket{1}$ and
  \begin{equation}
  \begin{aligned}
  \Hsplit|_\Sclock =\; &\ketbrab{1}_{\calA}\otimes\ketbrab{\wh1}_{\calC} + \ketbrab{0}_{\calA}\otimes\ketbrab{\wh2}_{\calC}\,+ \\
  &\ketbrab{\wh0}_{\calC} - \ketbra{\wh0}{\wh1}_{\calC} - \ketbra{\wh0}{\wh2}_{\calC}\,+ \\
  &\ketbrab{\wh1}_{\calC} - \ketbra{\wh1}{\wh0}_{\calC} + \ketbra{\wh1}{\wh2}_{\calC}\,+ \\
  &\ketbrab{\wh2}_{\calC} - \ketbra{\wh2}{\wh0}_{\calC} + \ketbra{\wh2}{\wh1}_{\calC}. \\
  \end{aligned}
  \end{equation}
  We compute $\corank(\Hsplit|_\Sclock)=2$ using SageMath \cite{sagemath} in the supplementary material \cite{sup}.
\end{proof}

\subsection{Single-qubit gate gadget}

Combining two split gadgets, we can implement a single-qubit unitary $U = \lambda_0\ketbrab{\psi_0}+\lambda_1\ketbrab{\psi_1}$ (see also \cref{fig:HU}):
\begin{equation}\label{eq:HU}
  \begin{aligned}
  H_U &= (\Hsplit)_{\calC_0\calC_1\calC_2} + (\Hsplit)_{\calC_5\calC_3\calC_4} + h_{1,3}(\lambda_0) + h_{2,4}(\lambda_1),  \\
  h_{i,j}(\lambda) &= (\ketbrab{10} + \ketbrab{01} - \lambda^*\ketbra{10}{01}-\lambda\ketbra{01}{10})_{\calC_i\calC_j}
  \end{aligned}
\end{equation}
where $(\Hsplit)_{\calC_5\calC_3\calC_4}$ means that we replace the register $\calC_0,\calC_1,\calC_2$ in $\Hsplit$ with $\calC_5,\calC_3,\calC_4$.

\tikzset{time/.style={circle,draw,inner sep=0mm,minimum size=6mm},every edge quotes/.style={sloped},cond/.style={thick,dashed},unitary/.style={thick,-Latex}}
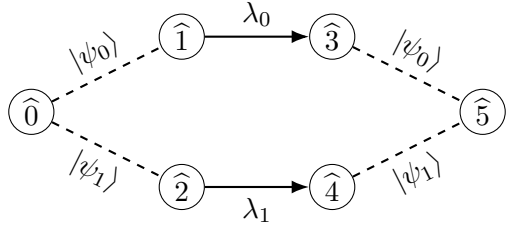
\begin{figure}[t]
  \centering
  \begin{tikzpicture}[]
    \node[time] (t0) at (0,0) {$\wh0$};
    \node[time] (t1) at (2,1) {$\wh1$};
    \node[time] (t2) at (2,-1) {$\wh2$};
    \node[time] (t3) at (4,1) {$\wh3$};
    \node[time] (t4) at (4,-1) {$\wh4$};
    \node[time] (t5) at (6,0) {$\wh5$};
    \draw[cond] (t0) edge["$\ket{\psi_0}$",above] (t1);
    \draw[cond] (t0) edge["$\ket{\psi_1}$",below] (t2);
    \draw[unitary] (t1) edge["$\lambda_0$",above] (t3);
    \draw[unitary] (t2) edge["$\lambda_1$",below] (t4);
    \draw[cond] (t3) edge["$\ket{\psi_0}$",above] (t5);
    \draw[cond] (t4) edge["$\ket{\psi_1}$",below] (t5);
  \end{tikzpicture}
  \caption{Graphical representation of the gadget $H_U$ defined in \cref{eq:HU}. Dashed edges indicate ``conditional transitions'', and the arrows indicate ``unitary transitions'' (which are of the form $\lambda I$ here), following the conventions of \cite{RGN24}.}\label{fig:HU}
\end{figure}

\begin{claim}\label{claim:HU}
  $H_U$ is Hermitian with $H_U|_\Sclock\succeq0$, $\Null(H_U|_\Sclock) = \{ \ket{\phi_0},\ket{\phi_1} \}$, where
  \begin{equation}
    \begin{aligned}
      \ket{\phi_0}&=\frac{1}{2}\ket{\psi_0}_\calA(\ket{\wh0} + \ket{\wh1} +\lambda_0\ket{\wh3} + \lambda_0\ket{\wh5})_\calC,\\
      \ket{\phi_1}&=\frac{1}{2}\ket{\psi_1}_\calA(\ket{\wh0} + \ket{\wh2} +\lambda_1\ket{\wh4} + \lambda_1\ket{\wh5})_\calC,
    \end{aligned}
  \end{equation} 
  and $H_U\ket{\phi_0} = H_U\ket{\phi_1} = 0$.
  $H_U$ satisfies the conditions (2) of \cref{lem:connect}.
\end{claim}
\begin{proof}
  As in \cref{claim:Hsplit}, observe that $\Pi_{\Sclock}H_U\Pi_{\Sclock}=H_U\Pi_{\Sclock}$.
  The analysis of the nullspace is similar to \cite[Lemma 3.2]{RGN24}.
  Let $\Pi_i = \ketbrab{\psi_i}_{\calA}\otimes I_{\calC}$ for $i\in\{0,1\}$.
  Then we have $H_U = \Pi_0H_U\Pi_0 + \Pi_1H_U\Pi_1$ with
  \begin{equation}\label{eq:PiHU}
    \begin{aligned}
      \Pi_0H_U\Pi_0|_\Sclock = \ketbrab{\psi_0}_{\calA}\otimes \bigl(&\ketbrab{\wh2} + \ketbrab{\wh4}\,+ \\
      &\ketbrab{\wh0} + \ketbrab{\wh1} - \ketbra{\wh0}{\wh1} -\ketbra{\wh1}{\wh0}\,+\\
      &\ketbrab{\wh1} + \ketbrab{\wh3} - \lambda_0^*\ketbra{\wh1}{\wh3} -\lambda_0\ketbra{\wh3}{\wh1}\,+\\
      &\ketbrab{\wh3} + \ketbrab{\wh5} - \ketbra{\wh3}{\wh5} -\ketbra{\wh5}{\wh3}\bigr)_{\calC}.
    \end{aligned}
  \end{equation}
  \Cref{eq:PiHU} now resembles the Kitaev's circuit Hamiltonian \cite{KSV02} with four timesteps and three gates $I,\lambda_0I, I$.
  Therefore $\Null(\Pi_iH_U\Pi_i|_\calC) = \Span\{\ket{\phi_i}\}$ for $i\in\{0,1\}$ ($i=1$ is analogous).

  It remains to verify the conditions of \cref{lem:connect}.
  (a) follows from the nullspace.
  For (b), define $L = 2\ketbra{\phi_0}{\psi_0} + 2\ketbra{\phi_1}{\psi_1}$.
  (c) and (d) are obvious.
  For (e) and (f), let $\ket{\alpha} = \alpha_0\ket{\psi_0} + \alpha_1\ket{\psi_1}$.
  We get $\bra{\wh0}_{\calC} L\ket{\alpha} = \ket{\alpha}$ and $\bra{\wh5}_{\calC} L\ket{\alpha} = \alpha_0\lambda_0\ket{\psi_0}+\alpha_1\lambda_1\ket{\psi_1} = U\ket{\alpha}$.
\end{proof}

For $U=\Tg_{2^k}$, the gadget $H_U$ is clearly in $\QQ(\zeta_{2^k})$.
For $U =\Hg$, this is less obvious.
$\Hg$ has eigenvalues $\pm1$ with eigenvectors
\begin{equation}
  \ket{\Hg_+} = \begin{pmatrix}
    a\\b
  \end{pmatrix},\quad\ket{\Hg_{-}} = \begin{pmatrix}
    -b\\
    a
  \end{pmatrix},\qquad
  a = \frac{\sqrt{2+\sqrt{2}}}{2}=\cos\frac\pi8,\quad b = \frac{\sqrt{2-\sqrt{2}}}{2}=\sin\frac\pi8,
\end{equation}
with $\ketbrab{\Hg_+},\ketbrab{\Hg_-}$ in $\QQ(\zeta_8)$, as 
\begin{equation}
    a^2 = \frac{2+\sqrt{2}}{4},\quad b^2 = \frac{2-\sqrt{2}}4,\quad ab = \frac{2}4.
\end{equation}
Hence, $H_\Hg$ in $\QQ(\zeta_8)$.

\subsection{Two-qubit gate gadget}

The last gadget we need is the CNOT gadget $H_\CXg$ (depicted in \cref{fig:HCX}), which acts on a logical register $\calA$ of two qubits $\calA_{0},\calA_1$ and a clock register $\calC$ of $12$ qubits $\calC_0,\dots,\calC_{11}$.
$H_{\CXg}$ is constructed by effectively nesting two instances of the single-qubit gadget.
Let $\Hsplit(\ket{\eta_0},\ket{\eta_1})_{\A_i\calC_u\calC_v\calC_t}$ denote an instance of $\Hsplit$ obtained by substituting $\ket{\eta_0},\ket{\eta_1}$ for $\ket{\psi_0},\ket{\psi_1}$, $\calA_i$ for $\calA$, and $\calC_{u},\calC_{v},\calC_t$ for $\calC_0,\calC_1,\calC_2$ in \cref{eq:HU}.
\begin{equation}\label{eq:HCX}
  \begin{aligned}
  H_{\CXg} =\; &\Hsplit(\ket0,\ket1)_{\calA_0\calC_0\calC_1\calC_2} + \Hsplit(\ket{+},\ket{-})_{\calA_1\calC_2\calC_4\calC_5}\,+\\
  &\Hsplit(\ket{+},\ket{-})_{\calA_1\calC_{10}\calC_7\calC_8} + \Hsplit(\ket0,\ket1)_{\calA_0\calC_{11}\calC_9\calC_{10}}\,+\\
  &h_{1,3}(1) + h_{3,6}(1) + h_{6,9}(1) + h_{4,7}(1) + h_{5,8}(-1)\,+\\
  &\ketbrab{11}_{\calA_0\calC_{3}} + \ketbrab{11}_{\calA_0\calC_{6}}
  \end{aligned}
\end{equation}

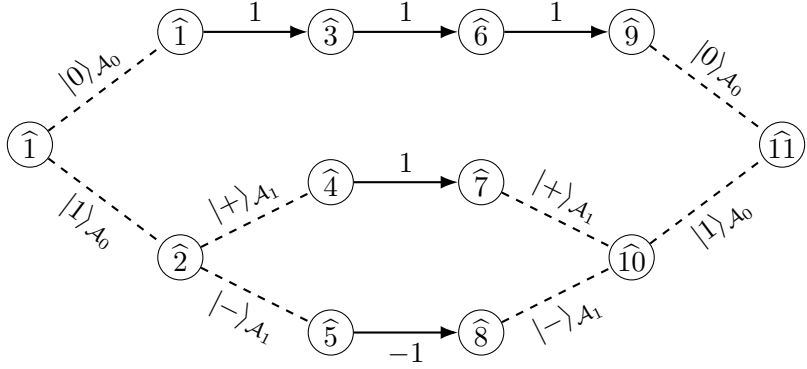
\begin{figure}[t]
  \centering
  \begin{tikzpicture}
    \node[time] (t0) at (-2,1.5) {$\wh1$};
    \node[time] (t1) at (0,3) {$\wh1$};
    \node[time] (t3) at (2,3) {$\wh3$};
    \node[time] (t6) at (4,3) {$\wh6$};
    \node[time] (t9) at (6,3) {$\wh9$};
    \node[time] (t11) at (8,1.5) {$\wh{11}$};

    \node[time] (t2) at (0,0) {$\wh2$};
    \node[time] (t4) at (2,1) {$\wh4$};
    \node[time] (t5) at (2,-1) {$\wh5$};
    \node[time] (t7) at (4,1) {$\wh7$};
    \node[time] (t8) at (4,-1) {$\wh8$};
    \node[time] (t10) at (6,0) {$\wh{10}$};

    \draw[cond] (t0) edge["$\ket{0}_{\calA_0}$",above] (t1);
    \draw[cond] (t0) edge["$\ket{1}_{\calA_0}$",below] (t2);
    \draw[unitary] (t1) edge["$1$",above] (t3);
    \draw[unitary] (t3) edge["$1$",above] (t6);
    \draw[unitary] (t6) edge["$1$",above] (t9);
    \draw[cond] (t9) edge["$\ket{0}_{\calA_0}$",above] (t11);
    \draw[cond] (t10) edge["$\ket{1}_{\calA_0}$",below] (t11);

    \draw[cond] (t2) edge["$\ket{+}_{\calA_1}$",above] (t4);
    \draw[cond] (t2) edge["$\ket{-}_{\calA_1}$",below] (t5);
    \draw[unitary] (t4) edge["$1$",above] (t7);
    \draw[unitary] (t5) edge["$-1$",below] (t8);
    \draw[cond] (t7) edge["$\ket{+}_{\calA_1}$",above] (t10);
    \draw[cond] (t8) edge["$\ket{-}_{\calA_1}$",below] (t10);
  \end{tikzpicture}
  \caption{Graphical representation of the gadget $H_\CXg$ defined in \cref{eq:HCX}.}\label{fig:HCX}
\end{figure}

\begin{claim}\label{claim:HCX}
  $H_\CXg$ is Hermitian with $\H_{\CXg}|_\Sclock\succeq0$, $\Null(H_\CXg|_\Sclock) = \{ \ket{\phi_{00}},\ket{\phi_{01}},\ket{\phi_{10}},\ket{\phi_{11}} \}$, where
  \begin{equation}
    \begin{aligned}
      \ket{\phi_{00}}&=\frac{1}{\sqrt6}\ket{0}_{\calA_0}\ket{+}_{\calA_1}(\ket{\wh0} + \ket{\wh1} + \ket{\wh3} + \ket{\wh6} + \ket{\wh9} + \ket{\wh{11}})_\calC,\\
      \ket{\phi_{01}}&=\frac{1}{\sqrt6}\ket{0}_{\calA_0}\ket{-}_{\calA_1}(\ket{\wh0} + \ket{\wh1} + \ket{\wh3} + \ket{\wh6} + \ket{\wh9} + \ket{\wh{11}})_\calC,\\
      \ket{\phi_{10}}&=\frac{1}{\sqrt6}\ket{1}_{\calA_0}\ket{+}_{\calA_1}(\ket{\wh0} + \ket{\wh2} + \ket{\wh4} + \ket{\wh7} + \ket{\wh{10}} + \ket{\wh{11}})_\calC,\\
      \ket{\phi_{11}}&=\frac{1}{\sqrt6}\ket{1}_{\calA_0}\ket{-}_{\calA_1}(\ket{\wh0} + \ket{\wh2} + \ket{\wh5} - \ket{\wh8} - \ket{\wh{10}} - \ket{\wh{11}})_\calC,\\
    \end{aligned}
  \end{equation} 
  and $H_\CXg\ket{\phi_{00}} = H_\CXg\ket{\phi_{01}} = H_\CXg\ket{\phi_{10}} = H_\CXg\ket{\phi_{11}} = 0$.
  $H_\CXg$ satisfies the conditions (2) of \cref{lem:connect}.
\end{claim}
\begin{proof}
  Analogous to \cref{claim:HU}.
\end{proof}

We verify the positivity and nullspaces of $H_\Hg,H_{\Tg},H_\CXg$ with SageMath in \cite{sup}.

\subsection{Assembling the Hamiltonian}

\begin{proof}[Proof of \cref{thm:2LH-complete}]
  Define $\Hprop$ as in the Nullspace Connection Lemma, with $H_1$ made up of the gate gadgets defined in \cref{eq:HU,eq:HCX}, and connected with $H_2$ defined in \cref{eq:hij}.
  By \cref{claim:projection-clock}, we can just apply \cref{lem:connect} to $\Hprop|_\Sclock$, since the lemma's requirements have been proven in \cref{claim:HU,claim:HCX}.
  Thus, we know that $\Hprop|_\Sclock\succeq0$, $\gamma(\Hprop|_\Sclock)\ge 1/\poly(n)$, and its nullspace is spanned by history states.
  Finally, we enforce correct ancilla initialization and an accepting computation with the projectors
  \begin{align}
    \Hin &= \sum_{i=1}^{n_1}\ketbrab{11}_{\calA_i\calC_0},\\
    \Hout &= \ketbrab{01}_{\calA_1\calC_{T}}.
  \end{align}
  We can complete the proof analogously to \cite[Theorem 3.1]{RGN24}.
\end{proof}

\section{Sparse Hamiltonian problems}\label{sec:sparse}

In this section, we show completeness results for sparse Hamiltonian problems.
We begin by extending \cref{lem:kELH} to sparse Hamiltonians.

\begin{lemma}\label{lem:ESH}
  $\ESH^{\QQ(\iu)} \in \QMA_{1}^{\calG_{2}}$ and $\ESH_\epsilon^{\QQ(\zeta_{2^k})} \in \QMA_{1,1-\epsilon'}^{\calG_{2^k}}$ with $\epsilon\le n^{-O(1)}$ and $\epsilon' \in\epsilon^{O(1)}$ for all $k\in\NN$, where $n$ is the number of qubits of the instance.
\end{lemma}

\subsection{Sparse Hamiltonian as linear combination of unitaries}

First, we need a technical lemma that extends \cite[Lemma 1]{KL21} to cyclotomic fields.

\begin{lemma}\label{lem:sparse-sum-of-unitaries}
  Let $H\in\ZZ[\zeta_{2^k}]^{2^n\times2^n}$ be an $n$-qubit $1$-sparse Hamiltonian with coefficients of at most $L$ bits.
  We can write $H=\sum_{l=0}^{L-1} \sum_{i=1}^{2^{k+1}}2^{l-1} U^{l,i}$ with $1$-sparse unitaries $U_{l,i}$, which can be efficiently implemented with the gateset $\calG_{2^k}$ for all $k\in\NN$.
\end{lemma}
\begin{proof}
  This proof is based on the proof of \cite[Lemma 1]{KL21} and \cite[Lemma 4.3]{BCCKS14}.
  Let $d=2^{k-1}$. $H$ has a unique decomposition
  \begin{equation}
    H = \sum_{l=0}^{L-1} 2^l\sum_{m=0}^{d-1} (C^{m,l} + D^{m,l}),
  \end{equation}
  such that each $C^{m,l}$ is $1$-sparse, Hermitian, only has entries $\pm\zeta_{2^k}^{m}$ above the diagonal, and $\pm\zeta_{2^k}^{-m}$ below the diagonal, and only zeros on the diagonal.
  Each $D^{m,l}$ is diagonal with entries $\pm\zeta_{2^k}^m$.
  Example in $\ZZ[\zeta_8]$:
  \begin{align}
    &\begin{pmatrix}
      1&\cdot&\cdot&\cdot\\
      \cdot&\sqrt{2}&\cdot&\cdot\\
      \cdot&\cdot&\cdot&1+\iu\\
      \cdot&\cdot&1-\iu&\cdot
    \end{pmatrix}
    = 
    \begin{pmatrix}
      1&\cdot&\cdot&\cdot\\
      \cdot&\zeta_8-\zeta_8^3&\cdot&\cdot\\
      \cdot&\cdot&\cdot&1+\zeta_8^2\\
      \cdot&\cdot&1-\zeta_8^2&\cdot
    \end{pmatrix}\nonumber\\
    &\qquad= 
    \begin{pmatrix}
      0&\cdot&\cdot&\cdot\\
      \cdot&0&\cdot&\cdot\\
      \cdot&\cdot&\cdot&1\\
      \cdot&\cdot&1&\cdot
    \end{pmatrix}
    +
    \begin{pmatrix}
      0&\cdot&\cdot&\cdot\\
      \cdot&0&\cdot&\cdot\\
      \cdot&\cdot&\cdot&\zeta_8^2\\
      \cdot&\cdot&\zeta_8^{-2}&\cdot
    \end{pmatrix}\\
    &\qquad\quad+\begin{pmatrix}
      1&\cdot&\cdot&\cdot\\
      \cdot&0&\cdot&\cdot\\
      \cdot&\cdot&0&\cdot\\
      \cdot&\cdot&\cdot&0
    \end{pmatrix}
    +
    \begin{pmatrix}
      0&\cdot&\cdot&\cdot\\
      \cdot&\zeta_8&\cdot&\cdot\\
      \cdot&\cdot&0&\cdot\\
      \cdot&\cdot&\cdot&0
    \end{pmatrix}
    +
    \begin{pmatrix}
      0&\cdot&\cdot&\cdot\\
      \cdot&-\zeta_8^3&\cdot&\cdot\\
      \cdot&\cdot&0&\cdot\\
      \cdot&\cdot&\cdot&0
    \end{pmatrix}
    \nonumber
  \end{align}
  Note that $C^{m,l}$ and $D^{m,l}$ are not yet unitary.
  Therefore, we decompose the $C^{m,l}=\frac12(C^{m,l,+}+C^{m,l,-})$ and $D^{m,l}=\frac12(D^{m,l,+}+D^{m,l,-})$ into unitaries by replacing the $0$'s in the example above with $+1$ and $-1$.
  Since $H$ is $1$-sparse, we can define for each row $i$, the index $j_i$ for the non-zero entry in row $i$.
  If row $i$ has no non-zero entry, we set $j_i=i$.
  The unitaries are defined as follows:
  \begin{equation}
    C_{ij_i}^{m,l,\pm} = \begin{cases}
      C^{m,l}_{ij_i}&\text{if }C^{m,l}_{ij_i}\ne 0\\
      \pm1&\text{if }C^{m,l}_{ij_i} =0\\
    \end{cases},
    \qquad
    D_{ii}^{m,l,\pm} = \begin{cases}
      D^{m,l}_{ii}&\text{if }D^{m,l}_{ii}\ne 0\\
      \pm1&\text{if } D^{m,l}_{ii}=0\\
    \end{cases}
  \end{equation}
  The $D^{m,l,\pm}$ are unitary since they are diagonal with only powers of $\zeta_{2^k}$ on the diagonal.
  We can write 
  \begin{equation}
    C^{m,l,\pm}=\sum_{i:i=j_i} C_{ii}^{m,l,\pm}\ketbrab{i} + \sum_{i:i\ne j_i}\left(C_{ij_i}^{m,l,\pm}\ketbra{i}{j_i} + (C_{ij_i}^{m,l,\pm})^*\ketbra{j_i}{i}\right),
  \end{equation}
  where $C_{ii}^{m,l,\pm}\in\{1,-1\}$ for $i=j_i$ (which implies $H_{ij_i}=C_{ij_i}^{m,l}=0$) and $\abs*{C_{ij_i}^{m,l,\pm}}=1$ for all $i$.
  Thus, the $C^{m,l,\pm}$ are also unitary.
  The efficient implementation of these unitaries is analogous to \cite{KL21}.
\end{proof}

\begin{proof}[Proof of \cref{lem:ESH}]
  Let $H$ be a $d$-sparse Hamiltonian on $n$ qubits ($d=\poly(n)$) in $\ZZ[\zeta_{2^k}]$ with coefficients of at most $L$ bits, i.e., $H$ has at most $d$ non-zero entries in each row and column.
  By \cite[Lemma 4.4]{BCCKS14}, we can write $H=\sum_{j=1}^{d^2}H_j$, where each $H_j$ is $1$-sparse and a query to any $H_j$ can be simulated with $O(1)$ queries to $H$.
  This decomposition assumes that the graph of $H$ is bipartite, which holds without loss of generality because $\Xg\otimes H$ has the same singular values as $H$.
  The idea is to construct a $d^2$-coloring of the edges of the graph of $G$, and letting each $H_j$ consist of the edges of color $j$.
  Hence, each $H_j$ is also in $\ZZ[\zeta_{2^k}]$ with coefficients of at most $L$ bits.
 
  Next, we prepare the $\QMAo$-verifier.
  By \cref{lem:sparse-sum-of-unitaries}, we can write 
  \begin{equation}\label{eq:H-sum-of-unitaries}
    H=\sum_{j=1}^{d^2}\sum_{l=0}^{L-1} \sum_{i=1}^{2^{k+1}}2^{l-1} U^{l,i}_j
  \end{equation}
  As in \cref{lem:kELH}, we prepare 
  \begin{equation}
    \ket{H}_{\calA} \propto \sum_{l=0}^{L-1}\sum_{j=1}^{d^2}\sum_{i=1}^{2^{k+1}}2^{l-1}\ket{i,j,l},
  \end{equation}
  and then apply the unitary $U_{\calA\calB} = \sum_{ijl} \ketbrab{i,j,l}_{\calA}\otimes (U^{l,i}_j)_{\calB}$ with the proof in register $\calB$.
  Soundness follows analogously to \cref{lem:kELH}, noting $L = O(\log(n))$ (see \cref{def:F-zeta}) and that flipping signs increases the norm by at most $\poly(n)$ due to sparsity.
\end{proof}

\subsection{QMA(2)}

It is straightforward to generalize \cref{thm:qma-g2,thm:cyclotomic-gateset} to $\QMA(2)$.
Simulating complex gates with real gates in the $\QMA(2)$ setting requires special consideration, for which we refer to \cite{McK13}.

\begin{theorem}\label{thm:qma(2)-g2}
  For any finite gateset $\calG$ in $\QQ(\iu)$, it holds that $\QMAo^{\calG}(2) \subseteq \QMAo^{\calG_2}(2)$.
\end{theorem}

\begin{theorem}\label{thm:qma(2)-cyclotomic-gateset}
  For any finite gateset $\calG$ in $\QQ(\zeta_{2^k})$ with $k\in\NN$, it holds that $\QMAo^{\calG}(2) \subseteq \QMAo^{\calG_{2^k}}(2)$.
\end{theorem}

Analogously to \cref{lem:ESH}, we get for $\QMAo(2)$:
\begin{lemma}\label{lem:ESSH}
  $\ESSH^{\QQ(\iu)} \in \QMA_{1}^{\calG_{2}}(2)$ and $\ESSH_\epsilon^{\QQ(\zeta_{2^k})} \in \QMA_{1,1-\epsilon'}^{\calG_{2^k}}(2)$ with $\epsilon\le n^{-O(1)}$ and $\epsilon' \in\epsilon^{O(1)}$ for all $k\in\NN$, where $n$ is the number of qubits of the instance.
\end{lemma}

\begin{theorem}\label{thm:QMA(2)}
  $\ESSH^{\QQ(\zeta_{2^k})}$ is $\QMAo^{\calG_{2^k}}(2)$-complete for all $k\in\NN$.
\end{theorem}
\begin{proof}
  Containment follows from \cref{lem:ESSH}, and hardness from \cite{CS12}.
\end{proof}

We note that the key difference to the containment proof in \cite{CS12}, is that we directly (try to) apply the sparse Hamiltonian directly to the proof, whereas \cite{CS12} uses Hamiltonian simulation (see \cite{LC17} for an asymptotically optimal algorithm), combined with phase estimation (see \cite{NC10}), which are both approximate techniques.
Furthermore, to achieve $L$ bits of precision using phase estimation, one needs to raise the unitary to the $2^L$-th power, or, in this case, simulate the Hamiltonian for time $2^L$, which generally takes exponential time \cite{LC17}.

Therefore, there was, to the best of our knowledge, no prior complete Hamiltonian problem for $\QMA(2)$ with a sub-polynomial promise gap (i.e. $n^{-\omega(1)}$).
We can now show a complete problem for every ``precision level'' of $\QMA(2)$, with the caveat that the promise gap is always close to $1$.
The general case is left open for future work.

\begin{theorem}\label{thm:AESSH}
  $\AESSH_\epsilon$ is complete for $\epsilon\mhyphen\QMA(2)\coloneq\bigcup_{c\in 1-\epsilon^{\omega(1)},s\in1-\epsilon^{O(1)}}\QMA_{c,s}(2)$ with $\epsilon \in n^{-\Omega(1)}$.
\end{theorem}
\begin{proof}
  For containment, observe that applying $H$ to the proof $\ket{\psi}$ in \cref{lem:ESH} succeeds with a probability of at least $\norm{H\ket{\psi}}/\poly(n)$ and at most $\norm{H\ket{\psi}}\cdot \poly(n)$.
  In the YES-case, the verifier therefore accepts with probability at least $1-\epsilon^{\omega(1)}$, and in the NO-case, it accepts with probability at most $1-\epsilon^{O(1)}$.

  Hardness follows again from \cite{CS12}, noting that the Hamiltonian $H$ constructed for an $N$-gate $\QMA_{c,s}(2)$-verifier is positive semidefinite, and a history state for a proof accepted with probability $c$, has energy $(1-c)/\poly(N)\le \epsilon^{\omega(1)}$.
  In the NO-case, we have $\braketb{\psi}{H}\ge (1-s)^{O(1)}/\poly(N)\ge\epsilon^{O(1)}$ for all $\ket{\psi}=\ket{\psi_1}\otimes\ket{\psi_2}$.
\end{proof}

\subsection{Clique homology}

\begin{theorem}\label{thm:GCH}
  $\GCH$ is $\QMAo^{\calG_2}$-complete.
\end{theorem}
\begin{proof}
  Containment follows from \cref{lem:sparse-sum-of-unitaries} since the Laplacian is sparse by \cref{lem:laplacian}.
  Hardness follows from a slight modification of the hardness proof of King and Kohler~\cite{KK24}.

\newcommand\pyth{\mathrm{Pyth}}
They show $\GCH\in\QMA$ and $\GCH$ is $\QMAo^{\calG_\pyth}$-hard, where $\calG_{\mathrm{Pyth}} = \{\CXg,U_\pyth\}$ with the Pythagorean gate $U_\pyth = \frac15\begin{psmallmatrix}
 3&4\\-4&3 
\end{psmallmatrix}$.
By \cref{thm:qma-g2}, $\QMAo^{\calG_\pyth} \subseteq \QMAo^{\calG_2}$, but we are not aware of nontrivial lower bounds for $\QMAo^{\calG_\pyth}$, since it is unclear to us how to perform classical computation with perfect completeness using $\calG_\pyth$, although the gateset is certainly universal.
It is straightforward to modify their construction to work with the gateset $\calG_2$.
We show in \cref{sec:4SAT-G2:CH} that our $4\hQSAT$ construction in \cref{thm:4SAT} for $\calG_2$ can be embedded into clique homology.
\end{proof}

\begin{theorem}\label{thm:CH}
  $\CH$ is $\PSPACE$-complete.
\end{theorem}
\begin{proof}
  The containment $\CH\in\PSPACE$ holds since we can decide whether the determinant of the Laplacian is zero in $\NC(\poly)=\PSPACE$ \cite{Csa75}.\footnote{We can even compute the rank of the Laplacian directly in $\NC(\poly)$ \cite{Chi85,Mul87}.}
  Hardness follows from the fact that $\QMA_1$ with exponentially small promise gap, denoted $\QMA_{1,1-1/\exp(n)}$, is equal to $\PSPACE$ \cite{Li22}.
  The $\QMA_{1,1-1/\exp(n)}$ verifier for $\PSPACE$ in \cite[Algorithm 8]{Li22} effectively just verifies the history state for a classical computation, and can therefore be implemented with $\calG_2$.
  Then we can just embed the corresponding $4\hQSAT$ instance into the clique homology problem as in \cref{thm:GCH}.
  By \cite[Theorem 10.1]{KK24}, we have for the Laplacian $\Delta^{2n-1}$, $\lmin(\Delta^{2n-1})=0$ iff $\lmin(H)=0$, where $n$ is the number of qubits of the Hamiltonian $H$.
\end{proof}

\section*{Acknowledgements}
I thank Tamara Kohler and Marcos Crichigno for helpful discussions, and especially Tamara Kohler for helping me understand Reference \cite{KK24}.
I also thank my supervisor Sevag Gharibian for his support and many discussions.
DR was supported by the DFG under grant number 432788384.

\appendix

\section{Omitted proofs}\label{sec:proofs}

\begin{proof}[Proof of \cref{lem:boundary}]
  See \cite[Eq. 6]{KK24} for \cref{eq:boundary sigma':a}.
  Note that $\partial^k\ket{\sigma'}$ only has support on those $\ket{\tau}$ with $\sigma\supset\tau\in\calK^{k-1}$.
  Then 
  \begin{equation}
    \partial^k\ket{\sigma'} = \sum_{j=0}^k \ketbrab{\sigma_{-j}'}(d^{k-1})^\dagger\ket{\sigma'} = \sum_{j=0}^k (-1)^{j}\cdot w(v_j)\ket{\sigma_{-j}'},
  \end{equation}
  as $\bra{\sigma_{-j}'}(d^{k-1})^\dagger\ket{\sigma'} = w(v_j)\braket*{\sigma_{-j}'}{[v_j] + \sigma_{-j}}=(-1)^{j}\cdot w(v_j)$ because moving $v_j$ in $[v_j] + \sigma_{-j}$ to the $j$-th position requires $j$ swaps with $0$ being the first position.
\end{proof}

\begin{proof}[Proof of \cref{lem:laplacian}]
  Recall $\Delta^k = d^{k-1}\partial^{k} + \partial^{k+1}d^k = (\partial^k)^\dagger\partial^k+ (d^k)^\dagger d^k$.
  If $\sigma=\tau$, then we have by \cref{eq:boundary sigma'}
  \begin{equation}
    \braketb{\sigma'}{\Delta^k} = \norm{\partial^k\ket{\sigma'}}^2 + \norm{d^k\ket{\sigma'}}^2 =  \sum_{v\in\sigma}w(v)^2 + \sum_{u\in \up(\sigma)} w(u)^2.
  \end{equation}
  If $\sigma$ and $\tau$ are not upper adjacent, then $d^k\ket{\sigma'}$ and $d^k\ket{\tau'}$ are orthogonal.
  Thus for the second and third case, we have
  \begin{equation}
    \bra{\sigma'}\Delta^k\ket{\tau'} = \left(\bra{\sigma'}(\partial^k)^\dagger\right) \ketbrab{\eta'} \left(\partial^k\ket{\tau'}\right) = \pm w(v_\sigma)w(v_\tau),
  \end{equation}
  since the common lower simplex $\eta$ is always unique.

  For the fourth case, consider the subcase that $\sigma$ and $\tau$ are upper adjacent.
  Since then $\abs{\sigma\cup\tau} = k+2$, we have $\abs{\sigma\cap\tau}=k$.
  Thus we can write (up to permutation) $\sigma=[v_\sigma,v_0,\dots,v_{k-1}]$ and $\tau=[v_\tau,v_0,\dots,v_{k-1}]$ (so that they have a similar lower complex).
  Thus,
  \begin{equation}
    \begin{aligned}
      \bra{\sigma'}\Delta^k\ket{\tau'} &= \bra{\sigma'}(\partial^k)^\dagger\partial^k\ket{\tau'} + \bra{\sigma'}(d^k)^\dagger d^k\ket{\tau'} \\
      &= w(v_\sigma)w(v_\tau) + w(v_\sigma)w(v_\tau)\braket*{[v_\tau,v_\sigma,v_0,\dots,v_{k-1}]'\vphantom{\big(}}{[v_\sigma,v_\tau,v_0,\dots,v_{k-1}]'} = 0.
    \end{aligned}
  \end{equation}
  Finally, if $\sigma$ and $\tau$ do not have a common lower complex and are not upper adjacent, then $d^k\ket{\sigma},d^k\ket{\tau}$ are orthogonal as well as $\partial^k\ket{\sigma},\partial^k\ket{\tau}$.
\end{proof}

\begin{proof}[Proof of \cref{prop:exact-synthesis-impossible}]
  Let $\calG=\{U_1,\dots,U_k\}$ with $U_i\in\QQ(\zeta_n)^{d\times d}$ for all $i=1,\dots,k$.
  We can write each $U_i$ as 
  \begin{equation}
    U_i = \sum_{j=0}^{n-1} \zeta_n^j \frac{A_{ij}}{a_{ij}}
  \end{equation}
  with $A_{ij}\in\ZZ^{d\times d}$ and $a_{ij}\in\NN$.
  Hence, we can also write any product of the $U_i$ (possibly acting on different qubits) as $\frac1b\sum_{j=0}^{n-1} \zeta_n^j B_j$
  with $B_{j}\in\ZZ^{d\times d}$ and $b$ a product of the $a_{ij}$.
  
  Therefore, it is clearly not possible to implement exactly all Pythagorean unitaries (this term was introduced in \cite{CK24}) of the form
  \begin{equation}
    \frac1c
    \begin{pmatrix}
      a&b\\-b&a
    \end{pmatrix}
  \end{equation}
  with $a^2+b^2=c^2$ and prime $c$, as there exist an infinite number of Pythagorean primes $c$.
\end{proof}

\section{Examples for simulation of cyclotomic gates}\label{sec:example-cyclotomic}

Below we show how $\sqrt{H}$ is simulated in \cref{lem:cyclotomic-simulation}, writing all matrix elements in the integral basis $\{1,\zeta_8,\zeta_8^2,\zeta_8^3\}$:

\begin{equation*}\def\arraystretch{1.5}
\left(\begin{array}{rr}
-\frac{1}{2} \zeta_{8}^{3} + \frac{1}{2} \zeta_{8}^{2} + \frac{1}{2} & -\frac{1}{2} \zeta_{8}^{3} \\
-\frac{1}{2} \zeta_{8}^{3} & \frac{1}{2} \zeta_{8}^{3} + \frac{1}{2} \zeta_{8}^{2} + \frac{1}{2}
\end{array}\right)\quad\xmapsto{\;\Psi\;}\quad
\left(\begin{array}{rrrr|rrrr}
\frac{1}{2} & \frac{1}{2} & -\frac{1}{2} & 0 & 0 & \frac{1}{2} & 0 & 0 \\
0 & \frac{1}{2} & \frac{1}{2} & -\frac{1}{2} & 0 & 0 & \frac{1}{2} & 0 \\
\frac{1}{2} & 0 & \frac{1}{2} & \frac{1}{2} & 0 & 0 & 0 & \frac{1}{2} \\
-\frac{1}{2} & \frac{1}{2} & 0 & \frac{1}{2} & -\frac{1}{2} & 0 & 0 & 0 \\
\hline
 0 & \frac{1}{2} & 0 & 0 & \frac{1}{2} & -\frac{1}{2} & -\frac{1}{2} & 0 \\
0 & 0 & \frac{1}{2} & 0 & 0 & \frac{1}{2} & -\frac{1}{2} & -\frac{1}{2} \\
0 & 0 & 0 & \frac{1}{2} & \frac{1}{2} & 0 & \frac{1}{2} & -\frac{1}{2} \\
-\frac{1}{2} & 0 & 0 & 0 & \frac{1}{2} & \frac{1}{2} & 0 & \frac{1}{2}
\end{array}\right)
\end{equation*}

\newpage
\noindent
Also for the QFT $F_8$ on $3$ qubits:

{
\scriptsize
\begin{gather*}\def\arraystretch{1.5}
  \left(\begin{array}{rrrrrrrr}
    -\frac{1}{4} \zeta_{8}^{3} + \frac{1}{4} \zeta_{8} & -\frac{1}{4} \zeta_{8}^{3} + \frac{1}{4} \zeta_{8} & -\frac{1}{4} \zeta_{8}^{3} + \frac{1}{4} \zeta_{8} & -\frac{1}{4} \zeta_{8}^{3} + \frac{1}{4} \zeta_{8} & -\frac{1}{4} \zeta_{8}^{3} + \frac{1}{4} \zeta_{8} & -\frac{1}{4} \zeta_{8}^{3} + \frac{1}{4} \zeta_{8} & -\frac{1}{4} \zeta_{8}^{3} + \frac{1}{4} \zeta_{8} & -\frac{1}{4} \zeta_{8}^{3} + \frac{1}{4} \zeta_{8} \\
    -\frac{1}{4} \zeta_{8}^{3} + \frac{1}{4} \zeta_{8} & \frac{1}{4} \zeta_{8}^{2} + \frac{1}{4} & \frac{1}{4} \zeta_{8}^{3} + \frac{1}{4} \zeta_{8} & \frac{1}{4} \zeta_{8}^{2} - \frac{1}{4} & \frac{1}{4} \zeta_{8}^{3} - \frac{1}{4} \zeta_{8} & -\frac{1}{4} \zeta_{8}^{2} - \frac{1}{4} & -\frac{1}{4} \zeta_{8}^{3} - \frac{1}{4} \zeta_{8} & -\frac{1}{4} \zeta_{8}^{2} + \frac{1}{4} \\
    -\frac{1}{4} \zeta_{8}^{3} + \frac{1}{4} \zeta_{8} & \frac{1}{4} \zeta_{8}^{3} + \frac{1}{4} \zeta_{8} & \frac{1}{4} \zeta_{8}^{3} - \frac{1}{4} \zeta_{8} & -\frac{1}{4} \zeta_{8}^{3} - \frac{1}{4} \zeta_{8} & -\frac{1}{4} \zeta_{8}^{3} + \frac{1}{4} \zeta_{8} & \frac{1}{4} \zeta_{8}^{3} + \frac{1}{4} \zeta_{8} & \frac{1}{4} \zeta_{8}^{3} - \frac{1}{4} \zeta_{8} & -\frac{1}{4} \zeta_{8}^{3} - \frac{1}{4} \zeta_{8} \\
    -\frac{1}{4} \zeta_{8}^{3} + \frac{1}{4} \zeta_{8} & \frac{1}{4} \zeta_{8}^{2} - \frac{1}{4} & -\frac{1}{4} \zeta_{8}^{3} - \frac{1}{4} \zeta_{8} & \frac{1}{4} \zeta_{8}^{2} + \frac{1}{4} & \frac{1}{4} \zeta_{8}^{3} - \frac{1}{4} \zeta_{8} & -\frac{1}{4} \zeta_{8}^{2} + \frac{1}{4} & \frac{1}{4} \zeta_{8}^{3} + \frac{1}{4} \zeta_{8} & -\frac{1}{4} \zeta_{8}^{2} - \frac{1}{4} \\
    -\frac{1}{4} \zeta_{8}^{3} + \frac{1}{4} \zeta_{8} & \frac{1}{4} \zeta_{8}^{3} - \frac{1}{4} \zeta_{8} & -\frac{1}{4} \zeta_{8}^{3} + \frac{1}{4} \zeta_{8} & \frac{1}{4} \zeta_{8}^{3} - \frac{1}{4} \zeta_{8} & -\frac{1}{4} \zeta_{8}^{3} + \frac{1}{4} \zeta_{8} & \frac{1}{4} \zeta_{8}^{3} - \frac{1}{4} \zeta_{8} & -\frac{1}{4} \zeta_{8}^{3} + \frac{1}{4} \zeta_{8} & \frac{1}{4} \zeta_{8}^{3} - \frac{1}{4} \zeta_{8} \\
    -\frac{1}{4} \zeta_{8}^{3} + \frac{1}{4} \zeta_{8} & -\frac{1}{4} \zeta_{8}^{2} - \frac{1}{4} & \frac{1}{4} \zeta_{8}^{3} + \frac{1}{4} \zeta_{8} & -\frac{1}{4} \zeta_{8}^{2} + \frac{1}{4} & \frac{1}{4} \zeta_{8}^{3} - \frac{1}{4} \zeta_{8} & \frac{1}{4} \zeta_{8}^{2} + \frac{1}{4} & -\frac{1}{4} \zeta_{8}^{3} - \frac{1}{4} \zeta_{8} & \frac{1}{4} \zeta_{8}^{2} - \frac{1}{4} \\
    -\frac{1}{4} \zeta_{8}^{3} + \frac{1}{4} \zeta_{8} & -\frac{1}{4} \zeta_{8}^{3} - \frac{1}{4} \zeta_{8} & \frac{1}{4} \zeta_{8}^{3} - \frac{1}{4} \zeta_{8} & \frac{1}{4} \zeta_{8}^{3} + \frac{1}{4} \zeta_{8} & -\frac{1}{4} \zeta_{8}^{3} + \frac{1}{4} \zeta_{8} & -\frac{1}{4} \zeta_{8}^{3} - \frac{1}{4} \zeta_{8} & \frac{1}{4} \zeta_{8}^{3} - \frac{1}{4} \zeta_{8} & \frac{1}{4} \zeta_{8}^{3} + \frac{1}{4} \zeta_{8} \\
    -\frac{1}{4} \zeta_{8}^{3} + \frac{1}{4} \zeta_{8} & -\frac{1}{4} \zeta_{8}^{2} + \frac{1}{4} & -\frac{1}{4} \zeta_{8}^{3} - \frac{1}{4} \zeta_{8} & -\frac{1}{4} \zeta_{8}^{2} - \frac{1}{4} & \frac{1}{4} \zeta_{8}^{3} - \frac{1}{4} \zeta_{8} & \frac{1}{4} \zeta_{8}^{2} - \frac{1}{4} & \frac{1}{4} \zeta_{8}^{3} + \frac{1}{4} \zeta_{8} & \frac{1}{4} \zeta_{8}^{2} + \frac{1}{4}
    \end{array}\right)\\[1mm]
  \text{\large$\xmapsto{\;\Psi\;}$}\\[1mm]
\def\arraystretch{1.5}\setlength\arraycolsep{1pt}
  \left(\begin{array}{rrrr|rrrr|rrrr|rrrr|rrrr|rrrr|rrrr|rrrr}
    0 & \frac{1}{4} & 0 & -\frac{1}{4} & 0 & \frac{1}{4} & 0 & -\frac{1}{4} & 0 & \frac{1}{4} & 0 & -\frac{1}{4} & 0 & \frac{1}{4} & 0 & -\frac{1}{4} & 0 & \frac{1}{4} & 0 & -\frac{1}{4} & 0 & \frac{1}{4} & 0 & -\frac{1}{4} & 0 & \frac{1}{4} & 0 & -\frac{1}{4} & 0 & \frac{1}{4} & 0 & -\frac{1}{4} \\
    \frac{1}{4} & 0 & \frac{1}{4} & 0 & \frac{1}{4} & 0 & \frac{1}{4} & 0 & \frac{1}{4} & 0 & \frac{1}{4} & 0 & \frac{1}{4} & 0 & \frac{1}{4} & 0 & \frac{1}{4} & 0 & \frac{1}{4} & 0 & \frac{1}{4} & 0 & \frac{1}{4} & 0 & \frac{1}{4} & 0 & \frac{1}{4} & 0 & \frac{1}{4} & 0 & \frac{1}{4} & 0 \\
    0 & \frac{1}{4} & 0 & \frac{1}{4} & 0 & \frac{1}{4} & 0 & \frac{1}{4} & 0 & \frac{1}{4} & 0 & \frac{1}{4} & 0 & \frac{1}{4} & 0 & \frac{1}{4} & 0 & \frac{1}{4} & 0 & \frac{1}{4} & 0 & \frac{1}{4} & 0 & \frac{1}{4} & 0 & \frac{1}{4} & 0 & \frac{1}{4} & 0 & \frac{1}{4} & 0 & \frac{1}{4} \\
    -\frac{1}{4} & 0 & \frac{1}{4} & 0 & -\frac{1}{4} & 0 & \frac{1}{4} & 0 & -\frac{1}{4} & 0 & \frac{1}{4} & 0 & -\frac{1}{4} & 0 & \frac{1}{4} & 0 & -\frac{1}{4} & 0 & \frac{1}{4} & 0 & -\frac{1}{4} & 0 & \frac{1}{4} & 0 & -\frac{1}{4} & 0 & \frac{1}{4} & 0 & -\frac{1}{4} & 0 & \frac{1}{4} & 0 \\
    \hline
     0 & \frac{1}{4} & 0 & -\frac{1}{4} & \frac{1}{4} & 0 & -\frac{1}{4} & 0 & 0 & -\frac{1}{4} & 0 & -\frac{1}{4} & -\frac{1}{4} & 0 & -\frac{1}{4} & 0 & 0 & -\frac{1}{4} & 0 & \frac{1}{4} & -\frac{1}{4} & 0 & \frac{1}{4} & 0 & 0 & \frac{1}{4} & 0 & \frac{1}{4} & \frac{1}{4} & 0 & \frac{1}{4} & 0 \\
    \frac{1}{4} & 0 & \frac{1}{4} & 0 & 0 & \frac{1}{4} & 0 & -\frac{1}{4} & \frac{1}{4} & 0 & -\frac{1}{4} & 0 & 0 & -\frac{1}{4} & 0 & -\frac{1}{4} & -\frac{1}{4} & 0 & -\frac{1}{4} & 0 & 0 & -\frac{1}{4} & 0 & \frac{1}{4} & -\frac{1}{4} & 0 & \frac{1}{4} & 0 & 0 & \frac{1}{4} & 0 & \frac{1}{4} \\
    0 & \frac{1}{4} & 0 & \frac{1}{4} & \frac{1}{4} & 0 & \frac{1}{4} & 0 & 0 & \frac{1}{4} & 0 & -\frac{1}{4} & \frac{1}{4} & 0 & -\frac{1}{4} & 0 & 0 & -\frac{1}{4} & 0 & -\frac{1}{4} & -\frac{1}{4} & 0 & -\frac{1}{4} & 0 & 0 & -\frac{1}{4} & 0 & \frac{1}{4} & -\frac{1}{4} & 0 & \frac{1}{4} & 0 \\
    -\frac{1}{4} & 0 & \frac{1}{4} & 0 & 0 & \frac{1}{4} & 0 & \frac{1}{4} & \frac{1}{4} & 0 & \frac{1}{4} & 0 & 0 & \frac{1}{4} & 0 & -\frac{1}{4} & \frac{1}{4} & 0 & -\frac{1}{4} & 0 & 0 & -\frac{1}{4} & 0 & -\frac{1}{4} & -\frac{1}{4} & 0 & -\frac{1}{4} & 0 & 0 & -\frac{1}{4} & 0 & \frac{1}{4} \\
    \hline
     0 & \frac{1}{4} & 0 & -\frac{1}{4} & 0 & -\frac{1}{4} & 0 & -\frac{1}{4} & 0 & -\frac{1}{4} & 0 & \frac{1}{4} & 0 & \frac{1}{4} & 0 & \frac{1}{4} & 0 & \frac{1}{4} & 0 & -\frac{1}{4} & 0 & -\frac{1}{4} & 0 & -\frac{1}{4} & 0 & -\frac{1}{4} & 0 & \frac{1}{4} & 0 & \frac{1}{4} & 0 & \frac{1}{4} \\
    \frac{1}{4} & 0 & \frac{1}{4} & 0 & \frac{1}{4} & 0 & -\frac{1}{4} & 0 & -\frac{1}{4} & 0 & -\frac{1}{4} & 0 & -\frac{1}{4} & 0 & \frac{1}{4} & 0 & \frac{1}{4} & 0 & \frac{1}{4} & 0 & \frac{1}{4} & 0 & -\frac{1}{4} & 0 & -\frac{1}{4} & 0 & -\frac{1}{4} & 0 & -\frac{1}{4} & 0 & \frac{1}{4} & 0 \\
    0 & \frac{1}{4} & 0 & \frac{1}{4} & 0 & \frac{1}{4} & 0 & -\frac{1}{4} & 0 & -\frac{1}{4} & 0 & -\frac{1}{4} & 0 & -\frac{1}{4} & 0 & \frac{1}{4} & 0 & \frac{1}{4} & 0 & \frac{1}{4} & 0 & \frac{1}{4} & 0 & -\frac{1}{4} & 0 & -\frac{1}{4} & 0 & -\frac{1}{4} & 0 & -\frac{1}{4} & 0 & \frac{1}{4} \\
    -\frac{1}{4} & 0 & \frac{1}{4} & 0 & \frac{1}{4} & 0 & \frac{1}{4} & 0 & \frac{1}{4} & 0 & -\frac{1}{4} & 0 & -\frac{1}{4} & 0 & -\frac{1}{4} & 0 & -\frac{1}{4} & 0 & \frac{1}{4} & 0 & \frac{1}{4} & 0 & \frac{1}{4} & 0 & \frac{1}{4} & 0 & -\frac{1}{4} & 0 & -\frac{1}{4} & 0 & -\frac{1}{4} & 0 \\
    \hline
     0 & \frac{1}{4} & 0 & -\frac{1}{4} & -\frac{1}{4} & 0 & -\frac{1}{4} & 0 & 0 & \frac{1}{4} & 0 & \frac{1}{4} & \frac{1}{4} & 0 & -\frac{1}{4} & 0 & 0 & -\frac{1}{4} & 0 & \frac{1}{4} & \frac{1}{4} & 0 & \frac{1}{4} & 0 & 0 & -\frac{1}{4} & 0 & -\frac{1}{4} & -\frac{1}{4} & 0 & \frac{1}{4} & 0 \\
    \frac{1}{4} & 0 & \frac{1}{4} & 0 & 0 & -\frac{1}{4} & 0 & -\frac{1}{4} & -\frac{1}{4} & 0 & \frac{1}{4} & 0 & 0 & \frac{1}{4} & 0 & -\frac{1}{4} & -\frac{1}{4} & 0 & -\frac{1}{4} & 0 & 0 & \frac{1}{4} & 0 & \frac{1}{4} & \frac{1}{4} & 0 & -\frac{1}{4} & 0 & 0 & -\frac{1}{4} & 0 & \frac{1}{4} \\
    0 & \frac{1}{4} & 0 & \frac{1}{4} & \frac{1}{4} & 0 & -\frac{1}{4} & 0 & 0 & -\frac{1}{4} & 0 & \frac{1}{4} & \frac{1}{4} & 0 & \frac{1}{4} & 0 & 0 & -\frac{1}{4} & 0 & -\frac{1}{4} & -\frac{1}{4} & 0 & \frac{1}{4} & 0 & 0 & \frac{1}{4} & 0 & -\frac{1}{4} & -\frac{1}{4} & 0 & -\frac{1}{4} & 0 \\
    -\frac{1}{4} & 0 & \frac{1}{4} & 0 & 0 & \frac{1}{4} & 0 & -\frac{1}{4} & -\frac{1}{4} & 0 & -\frac{1}{4} & 0 & 0 & \frac{1}{4} & 0 & \frac{1}{4} & \frac{1}{4} & 0 & -\frac{1}{4} & 0 & 0 & -\frac{1}{4} & 0 & \frac{1}{4} & \frac{1}{4} & 0 & \frac{1}{4} & 0 & 0 & -\frac{1}{4} & 0 & -\frac{1}{4} \\
    \hline
     0 & \frac{1}{4} & 0 & -\frac{1}{4} & 0 & -\frac{1}{4} & 0 & \frac{1}{4} & 0 & \frac{1}{4} & 0 & -\frac{1}{4} & 0 & -\frac{1}{4} & 0 & \frac{1}{4} & 0 & \frac{1}{4} & 0 & -\frac{1}{4} & 0 & -\frac{1}{4} & 0 & \frac{1}{4} & 0 & \frac{1}{4} & 0 & -\frac{1}{4} & 0 & -\frac{1}{4} & 0 & \frac{1}{4} \\
    \frac{1}{4} & 0 & \frac{1}{4} & 0 & -\frac{1}{4} & 0 & -\frac{1}{4} & 0 & \frac{1}{4} & 0 & \frac{1}{4} & 0 & -\frac{1}{4} & 0 & -\frac{1}{4} & 0 & \frac{1}{4} & 0 & \frac{1}{4} & 0 & -\frac{1}{4} & 0 & -\frac{1}{4} & 0 & \frac{1}{4} & 0 & \frac{1}{4} & 0 & -\frac{1}{4} & 0 & -\frac{1}{4} & 0 \\
    0 & \frac{1}{4} & 0 & \frac{1}{4} & 0 & -\frac{1}{4} & 0 & -\frac{1}{4} & 0 & \frac{1}{4} & 0 & \frac{1}{4} & 0 & -\frac{1}{4} & 0 & -\frac{1}{4} & 0 & \frac{1}{4} & 0 & \frac{1}{4} & 0 & -\frac{1}{4} & 0 & -\frac{1}{4} & 0 & \frac{1}{4} & 0 & \frac{1}{4} & 0 & -\frac{1}{4} & 0 & -\frac{1}{4} \\
    -\frac{1}{4} & 0 & \frac{1}{4} & 0 & \frac{1}{4} & 0 & -\frac{1}{4} & 0 & -\frac{1}{4} & 0 & \frac{1}{4} & 0 & \frac{1}{4} & 0 & -\frac{1}{4} & 0 & -\frac{1}{4} & 0 & \frac{1}{4} & 0 & \frac{1}{4} & 0 & -\frac{1}{4} & 0 & -\frac{1}{4} & 0 & \frac{1}{4} & 0 & \frac{1}{4} & 0 & -\frac{1}{4} & 0 \\
    \hline
     0 & \frac{1}{4} & 0 & -\frac{1}{4} & -\frac{1}{4} & 0 & \frac{1}{4} & 0 & 0 & -\frac{1}{4} & 0 & -\frac{1}{4} & \frac{1}{4} & 0 & \frac{1}{4} & 0 & 0 & -\frac{1}{4} & 0 & \frac{1}{4} & \frac{1}{4} & 0 & -\frac{1}{4} & 0 & 0 & \frac{1}{4} & 0 & \frac{1}{4} & -\frac{1}{4} & 0 & -\frac{1}{4} & 0 \\
    \frac{1}{4} & 0 & \frac{1}{4} & 0 & 0 & -\frac{1}{4} & 0 & \frac{1}{4} & \frac{1}{4} & 0 & -\frac{1}{4} & 0 & 0 & \frac{1}{4} & 0 & \frac{1}{4} & -\frac{1}{4} & 0 & -\frac{1}{4} & 0 & 0 & \frac{1}{4} & 0 & -\frac{1}{4} & -\frac{1}{4} & 0 & \frac{1}{4} & 0 & 0 & -\frac{1}{4} & 0 & -\frac{1}{4} \\
    0 & \frac{1}{4} & 0 & \frac{1}{4} & -\frac{1}{4} & 0 & -\frac{1}{4} & 0 & 0 & \frac{1}{4} & 0 & -\frac{1}{4} & -\frac{1}{4} & 0 & \frac{1}{4} & 0 & 0 & -\frac{1}{4} & 0 & -\frac{1}{4} & \frac{1}{4} & 0 & \frac{1}{4} & 0 & 0 & -\frac{1}{4} & 0 & \frac{1}{4} & \frac{1}{4} & 0 & -\frac{1}{4} & 0 \\
    -\frac{1}{4} & 0 & \frac{1}{4} & 0 & 0 & -\frac{1}{4} & 0 & -\frac{1}{4} & \frac{1}{4} & 0 & \frac{1}{4} & 0 & 0 & -\frac{1}{4} & 0 & \frac{1}{4} & \frac{1}{4} & 0 & -\frac{1}{4} & 0 & 0 & \frac{1}{4} & 0 & \frac{1}{4} & -\frac{1}{4} & 0 & -\frac{1}{4} & 0 & 0 & \frac{1}{4} & 0 & -\frac{1}{4} \\
    \hline
     0 & \frac{1}{4} & 0 & -\frac{1}{4} & 0 & \frac{1}{4} & 0 & \frac{1}{4} & 0 & -\frac{1}{4} & 0 & \frac{1}{4} & 0 & -\frac{1}{4} & 0 & -\frac{1}{4} & 0 & \frac{1}{4} & 0 & -\frac{1}{4} & 0 & \frac{1}{4} & 0 & \frac{1}{4} & 0 & -\frac{1}{4} & 0 & \frac{1}{4} & 0 & -\frac{1}{4} & 0 & -\frac{1}{4} \\
    \frac{1}{4} & 0 & \frac{1}{4} & 0 & -\frac{1}{4} & 0 & \frac{1}{4} & 0 & -\frac{1}{4} & 0 & -\frac{1}{4} & 0 & \frac{1}{4} & 0 & -\frac{1}{4} & 0 & \frac{1}{4} & 0 & \frac{1}{4} & 0 & -\frac{1}{4} & 0 & \frac{1}{4} & 0 & -\frac{1}{4} & 0 & -\frac{1}{4} & 0 & \frac{1}{4} & 0 & -\frac{1}{4} & 0 \\
    0 & \frac{1}{4} & 0 & \frac{1}{4} & 0 & -\frac{1}{4} & 0 & \frac{1}{4} & 0 & -\frac{1}{4} & 0 & -\frac{1}{4} & 0 & \frac{1}{4} & 0 & -\frac{1}{4} & 0 & \frac{1}{4} & 0 & \frac{1}{4} & 0 & -\frac{1}{4} & 0 & \frac{1}{4} & 0 & -\frac{1}{4} & 0 & -\frac{1}{4} & 0 & \frac{1}{4} & 0 & -\frac{1}{4} \\
    -\frac{1}{4} & 0 & \frac{1}{4} & 0 & -\frac{1}{4} & 0 & -\frac{1}{4} & 0 & \frac{1}{4} & 0 & -\frac{1}{4} & 0 & \frac{1}{4} & 0 & \frac{1}{4} & 0 & -\frac{1}{4} & 0 & \frac{1}{4} & 0 & -\frac{1}{4} & 0 & -\frac{1}{4} & 0 & \frac{1}{4} & 0 & -\frac{1}{4} & 0 & \frac{1}{4} & 0 & \frac{1}{4} & 0 \\
    \hline
     0 & \frac{1}{4} & 0 & -\frac{1}{4} & \frac{1}{4} & 0 & \frac{1}{4} & 0 & 0 & \frac{1}{4} & 0 & \frac{1}{4} & -\frac{1}{4} & 0 & \frac{1}{4} & 0 & 0 & -\frac{1}{4} & 0 & \frac{1}{4} & -\frac{1}{4} & 0 & -\frac{1}{4} & 0 & 0 & -\frac{1}{4} & 0 & -\frac{1}{4} & \frac{1}{4} & 0 & -\frac{1}{4} & 0 \\
    \frac{1}{4} & 0 & \frac{1}{4} & 0 & 0 & \frac{1}{4} & 0 & \frac{1}{4} & -\frac{1}{4} & 0 & \frac{1}{4} & 0 & 0 & -\frac{1}{4} & 0 & \frac{1}{4} & -\frac{1}{4} & 0 & -\frac{1}{4} & 0 & 0 & -\frac{1}{4} & 0 & -\frac{1}{4} & \frac{1}{4} & 0 & -\frac{1}{4} & 0 & 0 & \frac{1}{4} & 0 & -\frac{1}{4} \\
    0 & \frac{1}{4} & 0 & \frac{1}{4} & -\frac{1}{4} & 0 & \frac{1}{4} & 0 & 0 & -\frac{1}{4} & 0 & \frac{1}{4} & -\frac{1}{4} & 0 & -\frac{1}{4} & 0 & 0 & -\frac{1}{4} & 0 & -\frac{1}{4} & \frac{1}{4} & 0 & -\frac{1}{4} & 0 & 0 & \frac{1}{4} & 0 & -\frac{1}{4} & \frac{1}{4} & 0 & \frac{1}{4} & 0 \\
    -\frac{1}{4} & 0 & \frac{1}{4} & 0 & 0 & -\frac{1}{4} & 0 & \frac{1}{4} & -\frac{1}{4} & 0 & -\frac{1}{4} & 0 & 0 & -\frac{1}{4} & 0 & -\frac{1}{4} & \frac{1}{4} & 0 & -\frac{1}{4} & 0 & 0 & \frac{1}{4} & 0 & -\frac{1}{4} & \frac{1}{4} & 0 & \frac{1}{4} & 0 & 0 & \frac{1}{4} & 0 & \frac{1}{4}
    \end{array}\right)    
\end{gather*}
}
\newpage

\section{Toffoli from CS gate}\label{sec:toffoli-from-cs}

Construction of Toffoli gate from $\{\omega\Hg,\Sg,\CSg\}$ analogous to \cite[Figure 4.8]{NC10} using $\CSg^2 = \CZg, \CSg^3 = \CSg^\dagger, \CXg = (I\otimes\Hg)\CZg(I\otimes \Hg), \CCXg = (I\otimes I\otimes\Hg)\CCZg(I\otimes I\otimes\Hg)$:

\begin{center}
  \includegraphics{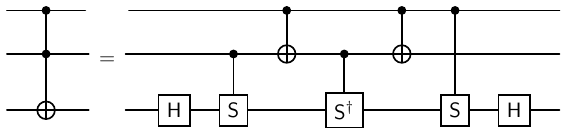}
\end{center}

\section{Quantum 4-SAT for \texorpdfstring{$\calG_2$}{G\_2}}\label{sec:4SAT-G2}

\newcommand{\Lu}{\textbf{u}}
\newcommand{\La}{\textbf{a}}
\newcommand{\Lai}{\textbf{a}_{\textbf{1}}}
\newcommand{\Laii}{\textbf{a}_{\textbf{2}}}
\newcommand{\Laiii}{\textbf{a}_{\textbf{3}}}
\newcommand{\Ld}{\textbf{d}}
\newcommand{\Ls}{\phantom{{}_{\textbf{3}}}}

The correctness of this construction follows from the Nullspace Connection Lemma (see \cref{lem:connect}).
We label the logical qu-$5$-its as $\Lu,\Lai,\Laii,\Laiii,\Ld$ (unborn, alive 1--3, dead) and implement them using $4$ physical qubits:

\begin{equation}\label{eq:logical}
  \begin{aligned}
    \ket{\Lu} &= \ket{1000}\\
    \ket{\Lai} &= \ket{0010}\\
    \ket{\Laii} &= \ket{0011}\\
    \ket{\Laiii} &= \ket{0001}\\
    \ket{\Ld} &= \ket{0100}
  \end{aligned}
\end{equation}

The remaining $4$-qubit states are penalized with $4$-local projectors. 
We can also easily enforce a clock space $\ket{C_1},\dots,\ket{C_T}$ of the form:
\newcommand\ketarray[1]{\ket{\;\;\begin{NiceArray}{w{l}{4mm}w{l}{4mm}w{l}{4mm}w{l}{4mm}w{c}{5mm}}#1\end{NiceArray}}}
\begin{align*}
  \ket{C_1} &= \ketarray{\Lu & \Lu & \Lu & \dotsm & \Lu}\\
  \ket{C_2} &= \ketarray{\Lai & \Lu & \Lu & \dotsm & \Lu}\\
  \ket{C_3} &= \ketarray{\Laii & \Lu & \Lu & \dotsm & \Lu}\\
  \ket{C_4} &= \ketarray{\Laiii & \Lu & \Lu & \dotsm & \Lu}\\
  \ket{C_5} &= \ketarray{\Ld & \Lu & \Lu & \dotsm & \Lu}\\
  \ket{C_6} &= \ketarray{\Ld & \Lai & \Lu & \dotsm & \Lu}\\
  \ket{C_7} &= \ketarray{\Ld & \Laii & \Lu & \dotsm & \Lu}\\
  \ket{C_8} &= \ketarray{\Ld & \Laiii & \Lu & \dotsm & \Lu}\\
  \ket{C_9} &= \ketarray{\Ld & \Ld & \Lu & \dotsm & \Lu}\\
  \ket{C_{10}} &= \ketarray{\Ld & \Ld & \Lai & \dotsm & \Lu}\\
  &\;\;\vdots\\
  \ket{C_{T}} &= \ketarray{\Ld & \Ld & \Ld & \dotsm & \Ld}\\
\end{align*}
Now we have $2$-local transitions between, e.g., $\ket{C_2},\ket{C_3},\ket{C_4}$, and $4$-local transitions between, e.g., $\ket{C_4}$ and $\ket{C_5}$.
Indexing the physical qubits from $1$ to $8$, we can implement the transitions
\begin{align}
  \ket{\Laiii\Lu}\leftrightarrow\ket{\Ld\Lu}\quad&\text{as}\quad \ketbraa{(\ket{0011}-\ket{1001})}_{2,3,4,5},\label{eq:a3u-du}\\
  \ket{\Ld\Lu}\leftrightarrow\ket{\Ld\Lai}\quad&\text{as}\quad \ketbraa{(\ket{1100}-\ket{1010})}_{2,5,7,8}\label{eq:du-da1}.
\end{align}
Thanks to \cref{lem:connect}, it now suffices to construct the gate gadgets on three timesteps $\ket{\Lai},\ket{\Laii},\ket{\Laiii}$.
See \cref{sec:2local} for a more detailed description of how to apply the lemma.
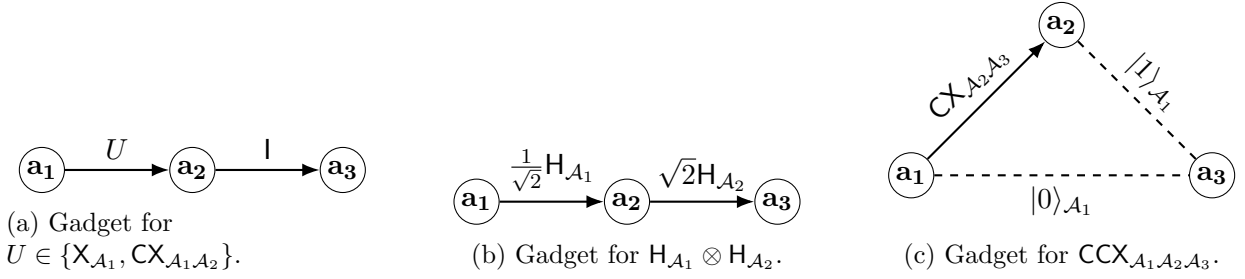
\begin{figure}[t]
  \centering
  \begin{subfigure}[b]{0.3\textwidth}
    \centering
    \begin{tikzpicture}[]
      \node[time] (a1) at (0,0) {$\Lai$};
      \node[time] (a2) at (2,0) {$\Laii$};
      \node[time] (a3) at (4,0) {$\Laiii$};
      \draw[unitary] (a1) edge["$U$",above] (a2);
      \draw[unitary] (a2) edge["$\Ig$",above] (a3);
    \end{tikzpicture}
    \caption{Gadget for\\$U\in \{\Xg_{\calA_1},\CXg_{\calA_1\calA_2}\}$.}
    \label{fig:4sat:cx}
  \end{subfigure}
  \hfill
  \begin{subfigure}[b]{0.3\textwidth}
    \begin{tikzpicture}[]
      \node[time] (a1) at (0,0) {$\Lai$};
      \node[time] (a2) at (2,0) {$\Laii$};
      \node[time] (a3) at (4,0) {$\Laiii$};
      \draw[unitary] (a1) edge["$\frac{1}{\sqrt{2}}\Hg_{\calA_1}$",above] (a2);
      \draw[unitary] (a2) edge["$\sqrt{2}\Hg_{\calA_2}$",above] (a3);
    \end{tikzpicture}
    \centering
    \caption{Gadget for $\Hg_{\calA_1}\otimes\Hg_{\calA_2}$.}
    \label{fig:4sat:hh}
  \end{subfigure}
  \hfill
  \begin{subfigure}[b]{0.3\textwidth}
    \begin{tikzpicture}[]
      \node[time] (a1) at (0,0) {$\Lai$};
      \node[time] (a2) at (2,2) {$\Laii$};
      \node[time] (a3) at (4,0) {$\Laiii$};
      \draw[unitary] (a1) edge["$\CXg_{\calA_2\calA_3}$",above] (a2);
      \draw[cond] (a2) edge["$\ket{1}_{\calA_1}$",above] (a3);
      \draw[cond] (a1) edge["$\ket{0}_{\calA_{1}}$",below] (a3);
    \end{tikzpicture}
    \centering
    \caption{Gadget for $\CCXg_{\calA_1\calA_2\calA_3}$.}
    \label{fig:4sat:ccx}
  \end{subfigure}
  \caption{Gadgets to implement $\calG_2$ with $4\hQSAT$. Arrows indicate unitary transitions, and dashed edges conditional transitions. The qubits of the computational register are denoted $\calA_1,\calA_2,\calA_3$.}
  \label{fig:4sat}
\end{figure}
The gadgets are graphically depicted in \cref{fig:4sat}, and formally defined below.
\begin{align}
  H_{\Xg} &= h_{12}(\Xg_{\calA_1}) + h_{23}(\Ig)\\
  H_{\CXg} &= h_{12}(\CXg_{\calA_1\calA_2}) + h_{23}(\Ig)\\
  H_{\Hg\otimes\Hg} &= h_{12}\left(\sqrt{2}\Hg_{\calA_1}\right) + h_{23}\left(\frac1{\sqrt2}\Hg_{\calA_2}\right)\\
  H_{\CCXg} &= h_{12}(\CXg_{\calA_2\calA_3}) + \ketbrab{0}_{\calA_1}\otimes h_{13}(\Ig) + \ketbrab{1}_{\calA_1}\otimes h_{23}(\Ig)\\
  h_{ij}(U) &= (I\otimes\ketbrab{\La_{\mathbf{i}}}_\calC + I\otimes\ketbrab{\La_{\mathbf{j}}} - U^\dagger \otimes \ketbra{\La_{\mathbf{i}}}{\La_{\mathbf{j}}} - U \otimes \ketbra{\La_{\mathbf{j}}}{\La_{\mathbf{i}}})_{\calA,\calC}\label{eq:4sat:hij}
\end{align}
The nullspaces of the gadgets are computed in the supplementary material \cite{sup}.

\subsection{Clique Homology}\label{sec:4SAT-G2:CH}

To embed the Hamiltonian into clique homology for \cref{thm:GCH}, we need to write all above gadgets $H_U= \sum_{l}\ketbrab{\phi_l}$, such that each $\ket{\phi_l}$ is proportional to an integer superposition of standard basis states.
Reference \cite{KK24} provides gadgets for $3$-qubit states and superpositions of at most two $4$-qubit states.\footnote{In fact, \cite{KK24} describes more general gadgets, but without correctness proof.}
Our gadgets are of the required form:
\begin{itemize}
  \item Projector onto the logical space (\cref{eq:logical}): $\ket{0000}$, $\ket{0101}$, $\ket{0110}$, $\ket{0111}$, $\ket{1001}$, $\ket{1010}$, $\ket{1011}$, $\ket{1100}$, $\ket{1101}$, $\ket{1110}$, $\ket{1111}$ on register $\calC=\calC_1\calC_2\calC_{3}\calC_{4}$
  \item \cref{eq:a3u-du}: $\ket{0011} - \ket{1001}$ on $\calC_{2}\calC_3\calC_4\calC_5$
  \item \cref{eq:du-da1}: $\ket{1100} - \ket{1010}$ on $\calC_{2}\calC_5\calC_7\calC_8$
  \item $h_{12}(\Xg_{\calA_1})$: $\ket{100} - \ket{111}, \ket{101} - \ket{110}$ on registers $\calC_{3}\calC_{4}\calA_1$
  \item $h_{23}(\Ig)$: $\ket{11} - \ket{01}$ on $\calC_{3}\calC_4$
  \item $h_{13}(\Ig)$: $\ket{10} - \ket{01}$ on $\calC_{3}\calC_4$
  \item $h_{12}(\CXg_{\calA_1\calA_2})$: $\ket{1000} - \ket{1100},\ket{1001} - \ket{1101},\ket{1010} - \ket{1111},\ket{1011} - \ket{1110}$ on $\calC_{3}\calC_4\calA_1\calA_2$
  \item $h_{12}(\sqrt{2}\Hg_{\calA_1})$: $\ket{100} - \ket{110} - \ket{111}, \ket{101} - \ket{100} + \ket{111}$ on $\calC_{3}\calC_4\calA_1$
  \item $h_{23}(\sqrt{1/2}\Hg_{\calA_1})$: $\ket{110} + \ket{111} - \ket{010}, \ket{110} - \ket{111} - \ket{011}$ on $\calC_{3}\calC_4\calA_1$
\end{itemize}
Note that each of the above states can be written as $\ket{x}\otimes\ket{\phi}$, where $\ket x$ is a standard basis state and $\ket{\phi}$ is a superposition of at most three standard basis states on two qubits.
We describe an alternative simplified construction for the gadgets of these states, following the procedure described in \cite[Section 8.3]{KK24}.
Let $\calK_1,\calK_2$ be simplicial complexes corresponding to gadgets for states $\ket{\phi_i}\in\CC^{2^{m_1}},\ket{\phi_2}\in\CC^{2^{m_1}}$, as defined in \cite[Section 8.2]{KK24}.
Then $\calK_i$ triangulates a $(2m_i-1)$-sphere and $f_i(\calK_i) = \calJ_i$, where $f_i$ is defined as in \cite[Section 8.2]{KK24} and $\calJ_i$ is the cycle corresponding to the state $\ket{\phi_i}$, for $i\in\{1,2\}$.
We observe that the join of the simplicial complexes $\calK_1$ and $\calK_2$, denoted $\calK = \calK_1*\calK_2$, triangulates the $(2(m_1+m_2)-1)$-sphere (see e.g. \cite{EP99}), and $f_1(f_2(\calK)) = \calJ_1*\calJ_2$, which is the cycle corresponding to $\ket{\phi_1}\otimes\ket{\phi_2}$ (see \cite[Section 8.1]{KK24}).
Hence, $\calK$ is a gadget for $\ket{\phi_1}\otimes\ket{\phi_2}$.

We prove the correctness of the concrete gadgets in the same way as \cite{KK24}, by computing the Euler characteristic and verifying that $f(\calK)=\calJ$.
Additionally, we compute the homology for the $2$-local gadgets algebraically using SageMath \cite{sagemath}.
For the gadgets implementing the projector onto a $2$-qubit state $\ket{\phi}$, we verify $\corank(\Delta^3)=3$, $\ket{\phi} \in \Image(\partial^4)$, and $\Kernel(\partial^3)/\Image(\partial^4)$ is spanned by the orthogonal complement of $\ket{\phi}$, where $\ket{\phi}$ both refers to the state as well as its representative in the clique complex.
We do the same verification numerically for the full $4$-qubit gadgets in C++, using \emph{Eigen} \cite{eigenweb} to represent vectors and matrices, \emph{Spectra} \cite{spectra} to compute eigenvalues of the Laplacian, and \emph{SuiteSparse} \cite{SPQR,AMD,COLAMD,CHOLMOD} to compute ranks of the boundary operators and verify $\ket{\phi}\in\Image(\partial^k)$.
We also compute the Euler characteristic in C++, using \emph{Cliquer} \cite{NO03} to compute the cliques of the gadget graphs.
All of our code is available as supplementary material \cite{sup}.

\printbibliography 

\end{document}